\documentclass[runningheads]{llncs}

\usepackage{amssymb,amsmath}
\usepackage{tikz}
\usepackage{tikzsymbols}
\usetikzlibrary{shapes,decorations.pathmorphing}
\usetikzlibrary{automata}

\usepackage[linesnumbered,lined,boxed,commentsnumbered]{algorithm2e}

\DeclareMathOperator{\Gain}{g}
\DeclareMathOperator{\tGain}{\tilde{\Gain}}

\DeclareMathOperator{\approxGain}{\tilde{\tilde{g}}}
\DeclareMathOperator{\allGain}{(\Gain_i)_{i\in\Pi}}
\newcommand{\allTar}{(F_i)_{i\in\Pi}}
\newcommand{\initG}{(\mathcal{G},v_0)}

\newcommand{\initX}{(\mathcal{X},x_0)}

\def\qGain#1{\tilde{\Gain}_{#1}}

\DeclareMathOperator{\Goal}{Goal}
\DeclareMathOperator{\Last}{Last}

\DeclareMathOperator{\Plays}{Plays}
\DeclareMathOperator{\Hist}{Hist}
\DeclareMathOperator{\Succ}{Succ}

\DeclareMathOperator{\Occ}{Occ}
\DeclareMathOperator{\Wit}{Wit}

\DeclareMathOperator{\Ar}{A}
\DeclareMathOperator{\tAr}{\tilde{\Ar}}
\DeclareMathOperator{\eAr}{X}

\DeclareMathOperator{\Aut}{\mathcal{A}}
\newcommand{\Clock}{C}
\newcommand{\Guard}{G}
\newcommand{\Loc}{L}
\newcommand{\Trans}{\rightarrow}
\newcommand{\vSet}{\Clock_V}
\newcommand{\timeAr}{\Ar_{\Aut}}

\newcommand{\Ordre}{\diamond}

\DeclareMathOperator{\Rel}{\mathbf{R}}

\def\Class#1{[#1]}
\def\outcome#1#2{\langle #1 \rangle_{#2}}

\newcommand{\initQ}{(\tilde{\mathcal{G}},\Class{v_0})}

\usepackage{graphicx}
\begin{document}
\title{On Subgame Perfect Equilibria in Turn-Based Reachability Timed Games\thanks{This work is supported by the ARC project ``Non-Zero Sum Game Graphs: Applications to Reactive Synthesis and Beyond" (F\'ed\'eration Wallonie-Bruxelles).}}
%
%
\author{Thomas Brihaye\inst{1} \and
Aline Goeminne\inst{1,2}}
\authorrunning{T. Brihaye and A. Goeminne}
%
 
\institute{Universit\'e de Mons (UMONS), Mons Belgium  \email{\{thomas.brihaye,aline.goeminne\}@umons.ac.be}\and
 Universit\'e libre de Bruxelles (ULB), Brussels, Belgium
}

\maketitle              
\begin{abstract}
We study multiplayer turn-based timed games with reachability objectives. In particular, we are interested in the notion of subgame perfect equilibrium (SPE). We prove that deciding the constrained existence of an SPE in this setting is EXPTIME-complete.

\keywords{multiplayer turn-based timed games  \and reachability objectives \and subgame perfect equilibria \and constrained existence problem}
\end{abstract}
%
%
%
%
%


\section{Introduction}

\subsubsection{Games} In the context of \emph{reactive systems},
\emph{two-player zero-sum games played on graphs} are commonly used to
model the purely antagonistic interactions between a system and its
environment~\cite{PnueliR89}. The system and the environment are the
two players of a game played on a graph whose vertices represent the
configurations. Finding how the system can ensure the achievement of
his objective amounts to finding, if it exists, a \emph{winning
  strategy} for the system.

When modeling complex systems with several agents whose objectives are
not necessarily antagonistic, the two-player zero-sum framework is too
restrictive and we rather rely on \emph{multiplayer non zero-sum
  games}. In this setting, the notion of winning strategy is replaced
by various notions of \emph{equilibria} including the famous concept
of \emph{Nash equilibrium} (NE)~\cite{nash50}. When considering games
played on graphs, the notion of \emph{subgame perfect equilibrium}
(SPE) is often preferred to the classical Nash
equilibrium~\cite{osbornebook}. Indeed, Nash equilibrium does not take
into account the sequential structure of the game and may allow
irrational behaviors in some subgames.

\subsubsection{Timed games}
\emph{Timed automata}~\cite{RajeevAlur.1994} is now a well established
model for complex systems including real time features. Timed automata
have been naturally extended into two-player zero-sum \emph{timed
  games}~\cite{AM99,CDFLL05,BCDFLL07,JT07}. Multiplayer non zero-sum extensions
have also been considered \cite{BouyerBM10,brenguier:tel-00827027}. In
these models both time and multiplayer aspects coexist. In this non
zero-sum timed framework, the main focus has been on NE, and, to our
knowledge, not on SPE.

\subsubsection*{Main contributions and organization of the paper.}
In this paper, we consider \emph{multiplayer, non zero-sum, turn-based
  timed games with reachability objectives} together with the concept
of SPE. We focus on the constrained existence problem (for SPE): given
a timed game, we want to decide whether there exists an SPE where some
players have to win and some other ones have to lose. The main result
of this paper is a proof that the SPE constrained existence problem is
EXPTIME-complete for reachability timed games. Let us notice that the
NE constrained existence problem for reachability timed games is also
EXPTIME-complete~\cite{brenguier:tel-00827027}. This may look
surprising as often, there is a complexity jump when going from NE to
SPE, for example the constrained existence problem on qualitative
reachability game is NP-complete for NE~\cite{CFGR16} and
PSPACE-complete for SPE~\cite{BrihayeBGR18}. Intuitively, the
complexity jump is avoided because the exponential blow up due to the
passage from SPE to NE is somehow absorbed by the classical
exponential blow up due to the classical region graph used for the
analysis of timed system.

In order to obtain an EXPTIME algorithm, we proceed in different
steps. In the first step, we prove that the game variant of the
\emph{classical region graph} is a good abstraction for the SPE
constrained existence problem. In fact, we identify conditions on
\emph{bisimulations} under which the study of SPE of a given
(potentially infinite game) can be reduced to the study of its
quotient. This is done in Section~\ref{section:SPE_game_quotient} for
(untimed) games with general objectives.  In
Section~\ref{section:reach}, we then focus on (untimed) finite
reachability game and provide an EXPTIME algorithm to solve the
constrained existence problem. Proving this result may look
surprising, as we already know from~\cite{BrihayeBGR18} that this
problem is indeed PSPACE-complete for (untimed) finite games. However
the PSPACE algorithm provided in~\cite{BrihayeBGR18} did not allow us
to obtain the EXPTIME algorithm for timed games. The latter EXPTIME
algorithm is discussed in Section~\ref{section:timedGames}.

\subsubsection*{Related works}
There are many results on SPEs played on graphs, we refer the reader
to \cite{Bruyere17} for a survey and an extended bibliography. Here we
focus on the results directly related to our contributions. The
constrained existence of SPEs is studied in finite multiplayer
turn-based games with different kinds of objectives, for example:
(qualitative) reachability and safety objectives~\cite{BrihayeBGR18},
$\omega$-regular winning conditions~\cite{Ummels06}, quantitative
reachability objectives~\cite{BrihayeBGRB19},... In~\cite{BouyerBM10},
they prove that the constrained existence problem for Nash equilibria
in concurrent timed games with reachability objectives is
EXPTIME-complete. This same problem in the same setting is studied
in~\cite{brenguier:tel-00827027} with others qualitative objectives.


\section{Preliminaries}
\subsubsection*{Transition systems, bisimulations and quotients} A \emph{transition system} is a tuple $T = ( \Sigma, V, E)$ where \emph{(i)}  $\Sigma$ is a finite alphabet; \emph{(ii)} $V$ a set of \emph{states} (also called \emph{vertices}) and \emph{(iii)} $E \subseteq V \times \Sigma \times V$ a set of \emph{transitions} (also called \emph{edges}). 
To ease the notation, an edge $(v_1,a,v_2)\in E$ is sometimes denoted by $v_1 \xrightarrow{a} v_2$.
Notice that $V$ may be uncontable. We said that the transition system is \emph{finite} if $V$ and $E$ are finite.\\

Given two transition systems on the same alphabet $T_1 = (\Sigma, V_1, E_1)$ and $T_2 = (\Sigma, V_2, E_2)$, a \emph{simulation of $T_1$ by $T_2$} is a binary relation $\Rel \subseteq V_1 \times V_2$ which satisfies the following conditions: \emph{(i)} $\forall v_1, v'_1 \in V_1$, $\forall v_2 \in V_2$ and $\forall a \in \Sigma$: $((v_1,v_2) \in \Rel \text{ and } v_1 \xrightarrow{a}_1 v'_1) \Rightarrow (\exists v'_2 \in V_2, v_2 \xrightarrow{a}_2 v'_2 \text{ and } (v'_1,v'_2) \in \Rel)$ and \emph{(ii)} for each $v_1 \in V_1$ there exists $v_2 \in V_2$ such that $(v_1, v_2 )\in \Rel$. We say that $T_2$ \emph{simulates} $T_1$. It implies that any transition $v_1 \xrightarrow{a}_1 v'_1$ in $T_1$ is simulated by a corresponding transition $v_2 \xrightarrow{a}_2 v'_2$ in $T_2$.

Given two transition systems on the same alphabet $T_1 = (\Sigma, V_1,
E_1)$ and $T_2 = (\Sigma, V_2, E_2)$, a \emph{bisimulation between
  $T_1$ and $T_2$} is a binary relation $\Rel \subseteq V_1 \times
V_2$ such that $\Rel$ is a simulation of $T_1$ by $T_2$ and the
converse relation $\Rel^{-1}$ is a simulation of $T_2$ by $T_1$ where
$\Rel^{-1} = \{(v_2,v_1) \in V_2 \times V_1 \mid (v_1,v_2) \in
\Rel\}$. When $\Rel$ is a bisimulation between two transition systems,
we write $\beta$ instead of $\Rel$. If $T = (\Sigma, V, E)$ is a
transition system, a bisimulation on $V \times V$ is called a
bisimulation on $T$.

Given a transition system $T = (\Sigma, V, E)$ and an equivalence
relation $\sim$ on $V$, we define the \emph{quotient of $T$ by
  $\sim$}, denoted by $\tilde{T} = (\Sigma, \tilde{V}, \tilde{E})$, as
follows: \emph{(i)} $\tilde{V} = \{ \Class{v}{\sim} \mid v \in V \}$
where $\Class{v}_{\sim} =\{ v' \in V \mid v \sim v'\}$
and \emph{(ii)} $\Class{v_1}_{\sim} \xrightarrow{a}_{\sim}
\Class{v_2}_{\sim}$ if and only if there exist $v'_1 \in
\Class{v_1}_{\sim}$ and $v'_2 \in \Class{v_2}_{\sim}$ such that $v'_1
\xrightarrow{a} v'_2$. When clear from the context which equivalence
relation is used, we write $[v]$ instead of $\Class{v}_{\sim}$.\\

Given a transition system $T = (\Sigma, V, E)$, a bisimulation $\sim$
on $T$ which is also an equivalence relation is called a
\emph{bisimulation equivalence}. In this context, the following result
holds. 

\begin{lemma}Given a transition system $T$ and a bisimulation equivalence $\sim$, there exists a bisimulation $\sim_q$ between $T$ and its quotient $\tilde{T}$. This bisimulation is given by the function $\sim_q: V \rightarrow \tilde{V}: v \mapsto \Class{v}_{\sim}$\end{lemma}

\subsubsection*{Turn-based games} \paragraph{Arenas, plays and histories} An \emph{arena} $\Ar = (\Sigma,V,E, \Pi,(V_i)_{i\in\Pi} )$ is a tuple where \emph{(i)} $T = (\Sigma, V,E)$ is a transition system such that for each $v \in V$, there exists $a \in \Sigma$ and $v'\in V$ such that $(v,a,v')\in E$; \emph{(ii)} $\Pi = \{1, \ldots, n\}$ is a finite set of players and \emph{(iii)} $(V_i)_{i\in\Pi}$ is a partition of $V$ between the players. An arena is finite if its transition sytem $T$ is finite.

A \emph{play} in $\Ar$ is an infinite path in its transition system,
\emph{i.e.,} $\rho = \rho_0\rho_1\ldots \in V^\omega$ is a play if for
each $i \in \mathbb{N}$, there exists $a \in \Sigma$ such that
$(\rho_i,a,\rho_{i+1})\in E$. A \emph{history} $h$ in $\Ar$ can be
defined in the same way but $h = h_0\ldots h_k \in V^*$ for some $k\in
\mathbb{N}$ is a finite path in the transition system. We denote the
set of plays by $\Plays$ and the set of histories by $\Hist$. When it
is necessary, we use the notation $\Plays_{\Ar}$ and $\Hist_{\Ar}$ to
recall the underlying arena $\Ar$. Moreover, the set $\Hist_i$ is the
set of histories such that their last vertex $v$ is a vertex of Player
$i$, \emph{i.e.,} $v \in V_i$. A play (resp. a history) in
$(\mathcal{G},v_0)$ is then a play (resp. a history) in $\mathcal{G}$
starting in $v_0$. The set of such plays (resp. histories) is denoted
by $\Plays(v_0)$ (resp. $\Hist(v_0)$). We also use the notation
$\Hist_i(v_0)$ when these histories end in a vertex $v \in V_i$.


Given a play $\rho \in \Plays$ and $k\in \mathbb{N}$, 
its suffix $\rho_k \rho_{k+1} \ldots $ is denoted by $\rho_{\geq k}$. 
We denote by $\Succ(v) = \{v' | (v,a, v') \in E \text{ for some }a\in \Sigma\}$ the set of successors of $v$, for $v \in V$ , and by $\Succ^*$ the transitive closure of $\Succ$. Given a play $\rho = \rho_0\rho_1\ldots $, the set $\Occ(\rho)=\{v \in V \mid \exists k, \rho_k = v\}$ is the set of vertices \emph{visited} along $\rho$.

\begin{remark}
	\label{rem:1}
When we consider a play in an arena $\Ar = (\Sigma,V,E, \Pi,(V_i)_{i\in\Pi})$,
we do not care about the alphabet letter associated with each edge of the play.
It is the reason why two different infinite paths in $T = (\Sigma, V, E)$
$v_0 \xrightarrow{a}v_1 \xrightarrow{a} \ldots \xrightarrow{a} v_n \xrightarrow{a} \ldots$
and $v_0 \xrightarrow{b}v_1  \xrightarrow{b} \ldots \xrightarrow{b} v_n \xrightarrow{b} \ldots$ correspond to only one play
$\rho = v_0v_1 \ldots v_n \ldots$ in $\Ar$. The same phenomenon appears with finite paths and histories.
We explain later why this is not a problem for our purpose.
\end{remark}

\paragraph{Multiplayer turn-based game} An \emph{(initialized multiplayer Boolean turn-based) game} is a tuple  $(\mathcal{G},v_0) = (\Ar, (\Gain_i)_{i\in\Pi} )$ such that: \emph{(i)} $\Ar = (\Sigma,V,E, \Pi,(V_i)_{i\in\Pi}) $ is an arena; \emph{(ii)} $v_0 \in V$ is the \emph{initial vertex} and \emph{(iii)} for each $i \in \Pi$, $\Gain_i: \Plays \rightarrow \{0,1\}$ is a gain function for Player~$i$. In this setting, each player $i \in \Pi$ is equipped with a set $\Omega_i \subseteq \Plays$ that we call the \emph{objective} of Player~$i$. Thus, for each $i \in \Pi$, for each $\rho \in \Plays$: $\Gain_i(\rho) = 1$ if and only if $\rho \in \Omega_i$. 
If $\Gain_i(\rho) = 1$ (resp. $=0$), we say that Player~$i$ \emph{wins} (resp. \emph{loses}) along $\rho$.
In the sequel of this document, we refer to the notion of initialized multiplayer Boolean turn-based game by the term ``game''. For each $\rho \in \Plays$, we write $\Gain(\rho)= p$ for some $p \in \{0,1\}^{|\Pi|}$ to depict $\Gain_i(\rho) = p_i$ for each $i \in \Pi$.

\paragraph{Strategies and outcomes} Given a game $(\mathcal{G},v_0)$, a \emph{strategy} of Player~$i$ is a function $\sigma_i: \Hist_i(v_0) \rightarrow V$ with the constraint that for each $hv \in \Hist_i(v_0)$, $\sigma_i(hv) \in \Succ(v). $ A play $\rho = \rho_0\rho_1\ldots$ is \emph{consistent} with $\sigma_i$ if for each $\rho_k$ such that $\rho_k \in V_i$, $\rho_{k+1} = \sigma_i(\rho_0\ldots \rho_k)$. A \emph{strategy profile} $\sigma = (\sigma_i)_{i \in \Pi}$ is a tuple of stratgies, one for each player. Given a game $(\mathcal{G},v_0)$ and a strategy profile $\sigma$, there exists a unique play from $v_0$ consistent with each strategy $\sigma_i$. We call this play the \emph{outcome} of $\sigma$ and denote it by $\outcome{\sigma}{v_0}$.

\begin{remark}
	\label{rem:2}
	We follow up Remark~\ref{rem:1}. The objectives we consider
        are of the form $\Omega \subseteq \Plays$. These objectives
        only depend on the sequence of visited states along a play
        (for example: visiting infinitely often a given state)
        regardless the sequence of visited alphabet letters. This is
        why defining the strategy of a player by a choice of the next
        vertex instead of a couple of an alphabet letter and a vertex
        is not a problem.  Actually, in all this paper one may
        consider that the alphabet is $\Sigma = \{a\}$. The reason why
        we allow alphabet on edge is to be able to consider
        \emph{synchronous products} of (timed)
        automata~\cite{Baier2008,RajeevAlur.1994}. In this way, we
        could consider wider class of objectives (see Section~\ref{section:thomas}).

\end{remark}
 
\subsubsection*{Subgame perfect equilibria}
In the multiplayer game setting, the solution concepts usually studied are \emph{equilibria} (see~\cite{GU08}). We here recall the concepts of Nash equilibrium and subgame perfect equilibrium.

Let $\sigma = (\sigma_i)_{i\in \Pi}$ be a strategy profile in a game $(\mathcal{G},v_0)$. When we highlight the role of Player~$i$, we denote $\sigma$ by $(\sigma_i, \sigma_{-i})$ where $\sigma_{-i}$ is the profile $(\sigma_j)_{j\in \Pi \setminus \{i\}}$. A strategy $\sigma'_i \neq \sigma_i$ is a \emph{deviating} strategy of Player~$i$, and it is a \emph{profitable deviation} for him if $\Gain_i(\outcome{\sigma}{v_0}) < \Gain_i(\outcome{\sigma'_i, \sigma_{-i}}{v_0})$. A strategy profile $\sigma$ in a game $(\mathcal{G},v_0)$ is a \emph{Nash equilibrium} (NE) if no player has an incentive to deviate unilaterally from his strategy, \emph{i.e.,} no player has a profitable deviation. 

A refinement of NE is the concept of \emph{subgame perfect equilibrium} (SPE) which is a strategy profile being an NE in each subgame. 
Formally, given a game $({\mathcal G},v_0) = (\Ar,(\Gain_i)_{i\in\Pi})$ and a history $hv \in \Hist(v_0)$, the game $(\mathcal{G}_{\restriction h},v)$ is called a \emph{subgame} of $(\mathcal{G},v_0)$ such that $\mathcal{G}_{\restriction h} = (\Ar, (\Gain_{i\restriction h})_{i\in \Pi})$ and $\Gain_{i\restriction h}(\rho) = \Gain_i(h\rho)$ for all $i \in \Pi$ and $\rho \in V^{\omega}$. Notice that $(\mathcal{G},v_0)$ is subgame of itself. Moreover if $\sigma_i$ is a strategy for Player~$i$ in $(\mathcal{G},v_0)$, then $\sigma_{i\restriction h}$ denotes the strategy in $(\mathcal{G}_{\restriction h},v)$ such that for all histories $h'\in \Hist_i(v)$, $\sigma_{i\restriction h}(h') = \sigma_i(hh')$. Similarly, from a strategy profile $\sigma$ in $(\mathcal{G},v_0)$, we derive the strategy profile $\sigma_{\restriction h}$ in $(\mathcal{G}_{\restriction h},v)$. Let $(\mathcal{G},v_0)$ be a game, following this formalism, a strategy profile $\sigma$ is a \emph{subgame perfect equilibrium} in $(\mathcal{G},v_0)$ if for all $hv \in \Hist(v_0)$, $\sigma_{\restriction h}$ is an NE in $(\mathcal{G}_{\restriction h},v)$.

\subsubsection*{Studied problem} Given a game $(\mathcal{G},v_0)$, several SPEs may coexist. It is the reason why we are interested in the \emph{constrained existence} of an SPE in this game: some players have to win and some other ones have to lose. The related decision problem is the following one:

\begin{definition}[Constrained existence problem]
	\label{def:constrained} Given a game $(\mathcal{G},v_0)$ and two gain profiles $x,y \in \{0,1\}^{|\Pi|}$, does there exist an SPE $\sigma$ in $(\mathcal{G},v_0)$ such that $x \leq \Gain(\outcome{\sigma}{v_0}) \leq y$.\end{definition}


\section{SPE in a game and its quotient}
\label{section:SPE_game_quotient}
In this section, we first define the concept of \emph{bisimulation between games} (resp. \emph{bisimulation on a game}). Then, we explain how given such bisimulations we can obtain a new game, called the \emph{quotient game}, thanks to a quotient of the initial game. Finally, we prove that if there exists an SPE in a game with a given gain profile, there exists an SPE in its associated quotient game with the same gain profile, and vice versa.

\subsection{Game bisimulation}

We extend the notion of bisimulation between transition systems (resp. on a transition system) to the one of bisimulation between games (resp. on a game). In this paper, by bisimulation between games (resp. on a game) we mean:

\begin{definition}[Game bisimulation] Given two games $(\mathcal{G},v_0) = (\Ar, (\Gain_i)_{i\in \Pi})$ and $(\mathcal{G}',v'_0) = (\Ar',(\Gain'_i)_{i\in \Pi})$ with the same alphabet and the same set of players, we say that $\sim \subseteq V \times V'$ is a \emph{bisimulation between $(\mathcal{G},v_0)$ and $(\mathcal{G}',v'_0)$ } if \emph{(i)} $\sim$ is a bisimulation between $T = (\Sigma, V, E)$ and $T' = (\Sigma, V',E')$ and \emph{(ii)} $v_0 \sim v'_0$. 
In the same way, if $\sim \subseteq V \times V$ we say that $\sim$ is a \emph{bisimulation on $(\mathcal{G},v_0)$} if $\sim$ is a bisimulation on $T = (\Sigma, V, E)$. 
\end{definition}

The notion of bisimulation equivalence on a transition system is extended in the same way to games. In the rest of this document, we use the following notations: \emph{(1)} If $\sim\subseteq V \times V'$ is a bisimulation between $(\mathcal{G},v_0) = (\Ar, (\Gain_i)_{i\in})$ and $(\mathcal{G}',v'_0) = (\Ar', (\Gain'_i)_{i\in})$, for each $\rho \in \Plays_{\Ar}$ and for all $\rho' \in \Plays_{\Ar'}$, we write $\rho \sim \rho'$ if and only if for each $n \in \mathbb{N}$: $\rho_n \sim \rho'_n$. \emph{(2)} If $\sim\subseteq V \times V$ is a bisimulation on $(\mathcal{G},v_0) = (\Ar, (\Gain_i)_{i\in})$ , for each $\rho \in \Plays_{\Ar}$ and for all $\rho' \in \Plays_{\Ar}$, we write $\rho \sim \rho'$ if and only if for each $n \in \mathbb{N}$: $\rho_n \sim \rho'_n$. \emph{(3)} Notations $1$ and $2$ can be naturally  adapted to histories\footnote{Once again, with this convention it is possible that two plays (or histories) such that $\rho \sim \rho'$ do not preserve the sequence of alphabet letters as it should be when we classicaly consider bisimulated paths in two bisimulated transitions systems. Remark~\ref{rem:2} explains why it is not a problem for us.}.\\


A natural property that should be satisfied by a bisimulation on a game is the respect of the vertices partition. It means that if a vertex bisimulates an other vertex, then these vertices should be owned by the same player. 

\begin{definition}[$\sim$ respects the partition] Given a game $(\mathcal{G},v_0) = (\Ar, (\Gain_i)_{i\in \Pi})$  and a bisimulation $\sim$ on $(\mathcal{G},v_0)$, we say that $\sim$ \emph{respects the partition} if for all $v,v' \in V$ such that $v\sim v'$, if $v \in V_i$ then $v' \in V_i$.
\end{definition}

\subsection{Quotient game}

Given a game $(\mathcal{G},v_0)$ and a bisimulation equivalence $\sim$ on it which respects the partition, one may consider its associated \emph{quotient game} $\initQ$ such that its transition system is defined as the quotient of the transition system of $(\mathcal{G},v_0)$.

\begin{definition}[Quotient game]
Given a game $(\mathcal{G},v_0) = (\Ar, (\Gain_i)_{i\in \pi})$ such that $\Ar =( \Sigma,V,E,\Pi, (V_i)_{i\in \Pi})$, if $\sim$ is a bisimulation equivalence on $(\mathcal{G},v_0)$ which respects the partition, the associated \emph{quotient game} $ \initQ = ( \tAr, (\tilde{\Gain}_i)_{i \in \Pi})$ is defined as follows: \emph{(i)} $\tAr = (\Sigma, \tilde{V}, \tilde{E}, (\tilde{V}_i)_{i\in \Pi})$ is such that $\tilde{T} = (\Sigma, \tilde{V}, \tilde{E})$ is the quotient of $T$ and, for each $i \in \Pi$, $\Class{v} \in \tilde{V}_i$ if and only if $v \in V_i$ and \emph{(ii)} for each $i \in \Pi$, $\tilde{\Gain}_i: \Plays_{\tAr} \rightarrow \{ 0, 1\}$ is the gain fonction of Player~$i$.
\end{definition}


In order to preserve some equivalent properties between a game and its quotient game, the equivalence relation on the game should respect the gain functions in both games. It means that if we consider two bisimulated plays either both in the game itself or one in the game and the other one in its quotient game, the gain profile of these plays should be equal.

\begin{definition}[$\sim$ respects the gain functions]
	Given an initialized game $(\mathcal{G},v_0) = (\Ar, (\Gain_i)_{i\in \pi})$ such that $\Ar =( \Sigma,V,E,\Pi, (V_i)_{i\in \Pi})$ and a bisimulation equivalence $\sim$ on $(\mathcal{G},v_0)$, we say that $\sim$ \emph{respects the gain functions} if the following properties hold: \emph{(i)} for each $\rho$ and $\rho'$ in $\Plays$, if $\rho \sim \rho'$ then $\Gain(\rho) = \Gain(\rho')$ and \emph{(ii)} for each $\rho \in \Plays_{\Ar}$ and $\tilde{\rho} \in \Plays_{\tAr}$, if $\rho \sim_q \tilde{\rho}$ then $\Gain(\rho) = \tilde{\Gain}(\tilde{\rho})$.
\end{definition}

\subsection{ Existence of SPE}

The aim of this section is to prove that, if there exists an SPE in a game equipped with a bisimilation equivalence which respects the partition and the gain functions, there exists an SPE in its associated quotient game with the same gain profile, and vice versa. 

\begin{theorem}
	\label{thm:eqSPE}
	Let $(\mathcal{G},v_0) = (\Ar, (\Gain_i)_{i\in \pi})$ be a game and   $\initQ = ( \tAr, (\tilde{\Gain}_i)_{i \in \Pi})$ its associated quotient game where $\sim$ is a bisimulation equivalence on $(\mathcal{G},v_0)$. If $\sim$ respects the partition and the gain functions, we have that: there exists an SPE $\sigma$ in  $(\mathcal{G},v_0)$ such that $\Gain(\outcome{\sigma}{v_0}) = p $ if and only if there exists an SPE $\tau$ in $\initQ$ such that $\tilde{\Gain}(\outcome{\tau}{\Class{v_0}}) = p$.	
\end{theorem}

The key idea is to prove that: if there exists an SPE in a game equipped with a bisimulation equivalence, there exists an SPE in this game which is \emph{uniform} and with the same gain profile. If $\sigma_i$ is an uniform strategy, each time we consider two histories $h\sim h'$, the choices of Player~$i$ taking into account $h$ or $h'$ are in the same equivalence class.

\begin{definition}
	Let  $(\mathcal{G},v_0)$ be a game and $\sim$ a bisimulation on it, we say that the strategy $\sigma_i$ is \emph{uniform} if for all $h,h'\in \Hist_i(v_0)$ such that $h \sim h'$, we have that $\sigma_i(h) \sim \sigma_i(h')$.
	A strategy profile $\sigma$ is $\emph{uniform}$ if for all $i \in \Pi$, $\sigma_i$ is uniform.
\end{definition}

\begin{proposition}
	\label{prop:uniformSPE}
	Let $(\mathcal{G},v_0) = (\Ar, (\Gain_i)_{i\in \pi})$ be a game and $\sim$ be a bisimulation equivalence on $(\mathcal{G},v_0)$ which respects the partition and such that for each $\rho$ and $\rho'$ in $\Plays$, if $\rho \sim \rho'$ then $\Gain(\rho) = \Gain(\rho')$, there exists an SPE $\sigma$ in $(\mathcal{G},v_0)$ such that $\Gain(\outcome{\sigma}{v_0}) = p$ if and only if there exists an SPE $\tau$ in $(\mathcal{G},v_0)$ which is uniform and such that $\Gain(\outcome{\tau}{v_0})= p.$
\end{proposition}


\section{Reachability games}
\label{section:reach}

In this section we focus on a particular kind of game called
\emph{reachability game}. In these games, each player has a subset of
vertices that he wants to reach.
First, we formally define the concepts of \emph{reachability games} and \emph{reachability quotient games}. Then, we provide an algorithm which solves the constrained existence problem in finite reachability games in time complexity at most exponential in the number of players and polynomial in the size of the transition system of the game.

\subsection{Reachability games and quotient reachability games}

\begin{definition} A \emph{reachability game} $(\mathcal{G},v_0) = (\Ar, (\Gain_i)_{i\in \Pi}, (F_i)_{i\in \Pi})$ is a game where each player $i \in \Pi$ is equipped with a \emph{target set} $F_i $ that he wants to reach. Formally, the objective of Player~$i$ is $\Omega_i = \{ \rho \in \Plays \mid \Occ(\rho) \cap F_i \neq \emptyset\}$ where $F_i \subseteq V$. This is a \emph{reachability objective}. 
\end{definition}

Given a reachability game $\initG = (\Ar, \allGain, \allTar)$ and a bisimulation equivalence $\sim$ on this game which respects the partition, one may consider its quotient game $\initQ = (\tAr, (\qGain{i})_{i \in \Pi}, (\tilde{F}_i)_{i\in \Pi})$ where for each $i \in \Pi$, $\tilde{F}_i  \subseteq \tilde{V}$. In attempts to ensure the respect of the gain functions by $\sim$, we add a natural property on $\sim$ (see Definition~\ref{def:respTar}) and define the sets $\tilde{F}_i$ in a proper way. In the rest of this paper, we assume that this property is satisfied and that the quotient game of a reachability game is defined as in Definition~\ref{def:quotientReach}.

\begin{definition}[$\sim$ respects the target sets]
	\label{def:respTar}
	Let $\initG $ be a reachability game and $\sim$ be a bisimulation equivalence on this game, we say that $\sim$ \emph{respects the target sets} if  for all $v \in V$ and for all $v' \in V$ such that $v\sim v'$:	$  v\in F_i \Leftrightarrow v'\in F_i)$. 
\end{definition}

\begin{definition}[Reachability quotient game]
	\label{def:quotientReach}\sloppy
	 Given a reachability  game $\initG= (\Ar, \allGain, \allTar)$ and a bisimulation equivalence $\sim$ on this game which respects the partition and the target sets,  its quotient game is the reachability game $\initQ = (\tAr, (\qGain{i})_{i \in \Pi}, (\tilde{F}_i)_{i\in \Pi})$ where $\tilde{F}_i = \{ \Class{v}_{\sim} \mid v \in F_i \}$ for each $i\in \Pi.$ We call this game the \emph{reachability quotient game}.
\end{definition}

\begin{lemma}
	Let $\initG$ be a reachability game  and let $\sim$ be  a bisimulation equivalence  which respects the target sets on this game, $\sim$ respects the gain fonctions.
\end{lemma}

\subsection{Complexity results}

 It is proved that the constrained existence problem is
 PSPACE-complete in finite reachability
 game~\cite{BrihayeBGR18}. Our final purpose is to obtain
   an EXPTIME algorithm for the constrained existence problem on
   \emph{reachability timed games} (see
   Section~\ref{section:timedGames}). Naively applying the PSPACE
   algorithm of~\cite{BrihayeBGR18} to the region games would lead to
   an EXPSPACE algorithm. That is why we provide here an alternative
   EXPTIME algorithm to solve the constrained existence problem on
   (untimed) finite games. This new algorithm will have the advantage
   to have a running time at most exponential only in the number of
   players (and polynomial in the size of its transition system). This
   feature will be crucial to obtain the EXPTIME algorithm on timed
   games.

\begin{theorem}
\label{thm:exptime}
Given a finite reachability game $\initG$, the
  constrained existence problem can be solved by an algorithm whose
  time complexity is at most exponential in $|\Pi|$ and polynomial in
  the size of its transition system.
\end{theorem}
This approach follows the proof for quantitative reachability games
in~\cite{BrihayeBGRB19}. This latter proof relies on two key
ingredients: \emph{(i)} the \emph{extended game} of a reachability
game and \emph{(ii)} an \emph{SPE outcome characterization based on a
  fixpoint computation of a labeling function of the states}. Those
two key ingredients will be defined below. Further technical details
can be found in~\cite{BrihayeBGRB19} for the quantitative case.\\

\subsubsection{Extended game} Let $\initG$ be finite a reachability game, its associated \emph{extended game} $\initX = (\eAr,(\Gain^X_i)_{i\in \Pi}, (F^X_i)_{i\in \Pi})$ is the reachability game such that the vertices are enriched with the set of players that have already visited their target sets along a history.  The arena $\eAr = (\Sigma,V^X,E^X, \Pi, (V^X_i)_{i \in \Pi})$ is defined as follows: \emph{(i)} $V^X = V \times 2^\Pi$; \emph{(ii)} $((v,I),a,(v',I')) \in E^X$ if and only if $(v,a,v')\in E$ and $I' = I \cup \{i \in \Pi \mid v'\in F_i\}$;  \emph{(iii)} $(v,I) \in V^X_i$ if and only if $v \in V_i$; \emph{(iv)} $(v,I) \in F^X_i$ if and only if $i\in I$ and \emph{(v)} $x_0 = (v_0,I_0)$ where $I_0 = \{ i \in \Pi \mid v_0 \in F_i \}$.

The construction of $\initX$ from $\initG$ causes an exponential blow-up  of the number of states. The main idea of this construction is that if you consider a play $\rho = (v_0,I_0)(v_1,I_1) \ldots (v_n,I_n) \ldots \in \Plays_{\eAr}(x_0)$ , the set $I_n$ means that each player~$i\in I_n$ has visited his target set along $\rho_0\ldots \rho_n$. The important points are that there is a one-to-one correspondance between plays in $\Plays_{\Ar}(v_0)$ and $\Plays_{\eAr}(x_0)$ and that the gain profiles of two corresponding plays beginning in the initial vertices are equal. From these observations, we have:

\begin{proposition}	
	\label{prop:eqExtended}
	Let $\initG$ be a reachability game and $\initX$ be its associated extended game, let $p \in \{0,1\}^{|\Pi|}$ be a gain profile, there exists an SPE $\sigma$ in $\initG$ with gain profile $p$ if and only if there exists an SPE $\tau$ in $\initX$ with gain profile $p$.
\end{proposition}

In the rest of this section, we will write $v \in V^X$ (instead of $(u,I)$) and we depict by $I(v)$ the set $I$ of the players who have already visited their target set.
\subsubsection{Outcome characterization} Once this extended game is build, we want a way to decide whether a play in this game corresponds to the outcome of an SPE or not: we want an \emph{SPE outcome characterization}. The vertices of the extended game are labeled thanks to a \emph{labeling function} $\lambda^*: V^X \rightarrow \{0,1\}$. For a vertex $v \in V^X$ such that $v \in V^X_i$, the value $1$  imposes that Player~$i$ should reach his target set if he follows an SPE from $v$  and the value $0$ does not impose any constraint on the gain of Player~$i$ from $v$.

The labeling function $\lambda^*$ is obtained thanks to an iterative procedure such that each step $k$ of the iteration provides a $\lambda^k$-labeling function. This procedure is based on the notion of \emph{$\lambda$-consistent play}: that is a play which sastifies the constraints given by $\lambda$ all along it.

\begin{definition}
	Let  $\lambda: V^X \rightarrow \{0,1\}$ be a labeling function and $\rho \in \Plays_{\eAr}$, we say that $\rho$ is \emph{$\lambda$-consistent} if for each $i \in \Pi$ and for each $n \in \mathbb{N}$ such that $\rho_n \in V^X_i$: $\Gain^X_i(\rho_{\geq n}) \geq \lambda(\rho_n)$. We write $\rho \models \lambda$.
\end{definition}

The iterative computation of the sequence $(\lambda^k)_{k\in \mathbb{N}}$ works as follows: \emph{(i)} at step $0$, for each $v \in V^X$, $\lambda^0(v) = 0$, \emph{(ii)} at step $k+1$, for each $v \in V^X$, by assuming that $v \in V^X_i$, $\lambda^{k+1}(v) = \max_{v'\in \Succ(v)} \min \{\Gain^X_i(\rho) \mid \rho \in \Plays_{\eAr}(v') \wedge \rho \models \lambda^k\}$ and \emph{(iii)} we stop when we find $n \in \mathbb{N}$  such that for each $v\in V^X$, $\lambda^{n+1}(v) = \lambda^{n}(v)$. The least natural number $k^*$ which satisfies $\emph{(iii)}$ is called the \emph{fixpoint} of $(\lambda^k)_{k\in\mathbb{N}}$ and $\lambda^*$ is defined as $\lambda^{k^*}$. The following lemma states that this natural number exists and so that the iterative procedure stops.

\begin{lemma}
	\label{lemma:fix}
	The sequence $(\lambda^k)_{k\in \mathbb{N}}$ reaches a fixpoint in $k^*\in \mathbb{N}$. Moreover, $k^*$ is at most equal to $|V|\cdot 2^{|\Pi|}$.
\end{lemma}
\begin{proof}[Proof sketch]
	In the initialization step, all the vertex values are equal to
        $0$. Then at each iteration, \emph{(i)} if the value of a
        vertex was equal to $1$ in the previous step, then it stays
        equal to $1$ all along the procedure and \emph{(ii)} if the
        value of the vertex was equal to $0$ then it either stays
        equal to $0$ (for this iteration step) or it becomes equal to
        $1$ (for all the next steps thanks to \emph{(i)}). At each
        step, at least one vertex value changes and when no value
        changes the procedure has reached a fixpoint which corresponds
        to the values of $\lambda^*$. Thus, it means that $\lambda^*$
        is obtained in at most $|V| \times 2^{|\Pi|}$ steps.
\end{proof}

As claimed in the following proposition, the labeling function $\lambda^*$ exactly caracterizes the set of SPE outcomes. The proof is quite the same as for the quantitative setting (\cite{BrihayeBGRB19}).

\begin{proposition}
	\label{prop:critOut}
	Let $\initX$ be the extended game of a finite reachability game $\initG$ and let $\rho^X \in \Plays_{\eAr}(x_0)$ be a play, 
	there exists an SPE $\sigma$ with outcome $\rho^X$ in $\initX$ if and only if $\rho^X$ is $\lambda^*$-consistent.
\end{proposition}

\subsubsection{Complexity}
Proposition~\ref{prop:critOut} allows us to prove
Theorem~\ref{thm:exptime}. Indeed, we only have to find a play in the
extended game which is $\lambda^*$-consistent and with a gain profile
which satisfies the constrained given by the decision problem.

\begin{proof}[Proof sketch of Theorem~\ref{thm:exptime}]
Let $\initG = (\Ar, \allGain, (F_i)_{i \in \Pi})$ be a reachability
game and let $\initX = (X, (\Gain^X_i)_{i\in \Pi}, (F^X_i)_{i\in
  \Pi})$ be its associated extended game. The game $\initX$ is build
from $\initG$ in time at most exponential in the number of players and
polynomial in the size of the transistion system of $\Ar$.

The proof will be organised in three steps whose respective proofs
will rely on the previous step(s): \emph{(i)} given a gain profile $p\in \{0,1\}^{|\Pi|}$, given $\mathcal{L}^k= \{
\lambda^k(v) \mid v \in V^X \}$ for some $k \in \mathbb{N}$ and given
some $v \in V^X$, we show that we can decide in the required
complexity the existence of a play which is $\lambda^k$-consistent,
beginning in $v$ and with gain profile $p$; \emph{(ii)} given
$\mathcal{L}^k$ for some $k \in \mathbb{N}$, we show that the
computation of $\mathcal{L}^{k+1}$ can be performed within the
required complexity; and finally \emph{(iii)} given $x,y \in
\{0,1\}^{|\Pi|}$, we show that the existence of a
$\lambda^*$-consistent play beginning in $x_0$ with a gain profile $p$
such that $x \leq p \leq y$ can be decided within the required
complexity.

\begin{itemize}
\item \textbf{Proof of $\emph{(i)}$}: Given $\mathcal{L}^k$, $v \in
  V^X$ and $p \in \{0,1\}^{|\Pi|}$, we want to know if there exists a
  play $\rho \in \Plays_{\eAr}(v)$ which is $\lambda^k$-consistent and
  with gain profile $p$. If a play $\rho$ is such that $\Gain^X(\rho)
  = p$, then for each $i \in \Pi$ such that $p_i = 1$, the condition
  of being a $\lambda^k$-consistent play is satisfied. For those such
  that $p_i =0$, for each $n \in \mathbb{N}$ such that $\rho_n \in
  V^X_i$, $\Gain^X_i(\rho_{\geq n}) = 0$ should be greater than
  $\lambda^k(\rho_n)$. This condition is satisfied if and only if for
  each $\rho_n \in V^X_i$, $\lambda^k(\rho_n) \neq 1$. Thus, we remove
  from $\initX$ all vertices (and all related edges) $v\in V^X_i$ such
  that $\lambda^k(v)=1$, for each player~$i$ such that $p_i = 0$.
  Then, we only have to check if there exists a play $\rho$ which
  begin in $v$ and with gain profile $p$ in this modified extended
  reachability game. This can be done in $O(2^{|\Pi|}\cdot (|V^X|
  +|E^X|))$ (\cite[Lemma 23]{BrihayeBGR18}), thus this procedure runs
  in time at most exponential in the number of players and polynomial
  in the size of the transition system of $\Ar$.
		
\item\textbf{Proof of $\emph{(ii)}$}: Given $\mathcal{L}^k$, we want
  to compute $\mathcal{L}^{k+1}$. For each $v \in V^X$,
  $\lambda^{k+1}(v) = \max_{v'\in\Succ(v)} \min\{\Gain^X_i(\rho) \mid
  \rho \in \Plays_{\eAr}(v') \wedge \rho \models \lambda^k \}$ (by
  assuming that $v \in V^X_i$). Thus for each $v' \in \Succ(v)$, we
  have to compute $min = \min\{\Gain^X_i(\rho) \mid \rho \in
  \Plays_{\eAr}(v') \wedge \rho \models \lambda^k \}$. But $min = 0$
  if and only if there exists $\rho \in \Plays_{\eAr}(v')$ which is
  $\lambda^k$-consistent and such that $\Gain^X_i(\rho)=0$. Thus for
  each $p \in \{0,1\}^{|\Pi|}$ such that $p_i =0$, we use point
  $\emph{(i)}$ to decide if $min =0$. From that follows a procedure
  which runs in $O(|V^X|\cdot |V^X| \cdot 2^{|\Pi|} \cdot
  2^{|\Pi|}\cdot (|V^X| +|E^X|))$ (running time at most exponential in
  the number of players and polynomial in the size of the transition
  system $\Ar$). 
		
\item\textbf{Proof of $\emph{(iii)}$}: It remains to prove that the
  existence of a $\lambda^*$-consistent play beginning in $x_0$ with a
  gain profile $p$ such that $x \leq p \leq y$ can be decided within
  the required complexity. In order to do so, we evaluate the
  complexity to obtain $\lambda^*$. First, we build $\mathcal{L}^0
  = \{ \lambda^0(v) = 0 \mid v \in V^X\}$ in $O(|V^X|)$ time. Then, we
  apply point \emph{(ii)} at most $|V|\cdot 2^{|\Pi|}$ times (by
  Lemma~\ref{lemma:fix}) to obtain $\mathcal{L}^*$. Given $x,y \in
  \{0,1\}^{|\Pi|}$, we consider each $p \in \{0,1\}^{|\Pi|}$ such that
  $x\leq p\leq y$ (at most $2^{|\Pi|}$ such ones) and we use point
  \emph{(i)} to check if there exists a play which begins in $x_0$
  with gain profile $p$ and which is $\lambda^*$-consistent. This can
  be done in running time at most exponential in the number of players
  and polynomial in the size of the transition system of $\Ar$.
\end{itemize}
We conclude the proof by applying Proposition~\ref{prop:critOut}.
\end{proof}

\section{Application to Timed Games}
\label{section:timedGames}

In this section, we are interested in models which are enriched with
\emph{clocks} and \emph{clock guards} in order to consider time
elapsing. \emph{Timed automata} \cite{RajeevAlur.1994} are well known
among such models. We recall some of their classical concepts, then we
explain how \emph{(turn-based) timed games} derive from timed
automata.

\subsection{Timed automata and timed games}

In this section, we use the following notations.  The set $\Clock = \{
c_1, \ldots, c_k \}$ denotes a set of $k$ \emph{clocks}.  A
\emph{clock valuation} is a function $\nu : \Clock \rightarrow
\mathbb{R}^+$. The set of clock valuation is depicted by $\vSet$.
Given a clock valuation $\nu$, for $i \in \{1,\ldots, k\}$, we
sometimes write $\nu_i$ instead of $\nu(c_i)$.  Given a clock
valuation $\nu$ and $d \in \mathbb{R}^+$, $\nu + d$ denote the clock
valuation $\nu+d: \Clock \rightarrow \mathbb{R}^+$ such that
$(\nu+d)(c_i) = \nu(c_i) + d$ for each $c_i \in \Clock$.  A
\emph{guard} is any finite conjunctions of expressions of the form $c_i
\Ordre x $ where $c_i$ is a clock, $x \in \mathbb{N}$ is a natural
number and $\Ordre$ is one of the symbols $\{ \leq, <, = , > , \geq
\}$.  We denote by $\Guard$ the set of guards. Let $g$ be a guard and
$\nu$ be a clock valuation, notation $\nu \models g$ means that
$(\nu_1, \ldots, \nu_k)$ satisfies $g$.  A \emph{reset} $Y \in
2^{\Clock}$ indicates which clocks are reset to $0$.  We denote by
$[Y\leftarrow 0]\nu$ the valuation $\nu'$ such that for each $c \in
Y$, $\nu'(c) = 0$ and for each $c \in \Clock\backslash Y$, $\nu'(c) =
\nu(c)$.\\

A \emph{timed automaton} (TA) is a tupple $(\Aut,\ell_0) = (\Sigma,\Loc,\Trans,\Clock)$ where: \emph{(i)} $\Sigma$ is a finite alphabet;
		\emph{(ii)} $\Loc$ is a finite set of \emph{locations};
		\emph{(iii)} $\Clock$ is a finite set of \emph{clocks};
		\emph{(iv)} $\Trans \subseteq \Loc \times \Sigma \times \Guard \times 2^{\Clock} \times \Loc$ a finite set of transitions;
		and \emph{(v)} $\ell_0 \in \Loc$ an initial location.
Additionnaly, we may equipped a timed automaton with a set of players and partition the locations between them. It results in a \emph{players partitioned timed automaton}.

\begin{definition}[(Reachability) Players partitioned timed automaton]
		A \emph{players partioned timed automaton} (PPTA) $(\Aut,\ell_0) = (\Sigma,\Loc,\Trans,\Clock, \Pi, (\Loc_i)_{i\in \Pi})$ is a timed automaton equipped with: \emph{(i)} $\Pi$ a finite set of players and \emph{(ii)} $(\Loc_i)_{i\in \Pi}$ a partition of the locations between the players.
		
		If $(\Aut,\ell_0)$ is equipped with a target set $\Goal_i \subseteq \Loc$ for each player $i\in \Pi$, we call it a \emph{reachability} PPTA.
\end{definition}

The semantic of a timed automaton $(\Aut,\ell_0)$ is given by its associated transition system $T_{\Aut} = (\Sigma, V, E)$ where:
\emph{(i)} $V = \Loc \times \vSet $ is a set of vertices of the form $(\ell,\nu)$ where $\ell$ is a location and $\nu: \Clock \rightarrow \mathbb{R}^+$ is a clock valuation;
		and \emph{(ii)} $E \subseteq V \times \Sigma \times V$ is such that $((\ell, \nu), a, (\ell',\nu')) \in E$ if $(\ell,a, g, Y, \ell')\in \Trans$ for some $g \in \Guard$ and some $Y \in 2^{\Clock}$, and there exists $d \in \mathbb{R}^+$ such that: \emph{(1)} for each $x \in X\backslash Y$: $\nu'(x) = \nu(x)+ d$ \textbf{(time elapsing)}; \emph{(2)} for each $x \in Y$: $\nu'(x) = 0 $ \textbf{(clocks resetting)}; \emph{(3)} $\nu + d \models g$ \textbf{(respect of the guard)}.


In the same way, the semantic of a PPTA $(\Aut,\ell_0)$ is given by its associated game $(\mathcal{G}_{\Aut},v_0)$.
\begin{definition}[(Reachability) Timed games $\mathcal{G}_{\Aut}$]
	Let $(\Aut,\ell_0) = (\Sigma,\Loc,\Trans,\Clock, \Pi, (\Loc_i)_{i\in \Pi})$ be a PPTA, its associated game $(\mathcal{G}_{\Aut},v_0) = (\timeAr, \allGain)$, called \emph{timed game}, is such that: \emph{(i)} $\timeAr = (\Sigma, V, E, \Pi, (V_i)_{i\in \Pi})$ where $T_{\Aut} = (\Sigma, V, E)$ is the associated transition system of $(\Aut,\ell_0)$ and, for each $i \in \Pi$, $(\ell, \nu) \in V_i$ if and only if $\ell \in \Loc_i$; \emph{(ii)} for each $i\in \Pi$, $\Gain_i: \Plays_{\timeAr} \rightarrow \{0,1\}$ is a gain function; \emph{(iii)} $v_0 = (\ell_0, \mathbf{0})$ where $\mathbf{0}$ is the clock valuation such that for all $c\in \Clock$, $\mathbf{0}(c) = 0 $.
	
		
		If $(\Aut,\ell_0)$ is a reachability PPTA, its associated timed game is a reachability game $(\mathcal{G}_{\Aut},v_0) = (\timeAr, \allGain, (F_i)_{i\in \Pi})$ such that for each $i\in \Pi$, $(\ell,\nu) \in F_i$ if and only if $\ell \in \Goal_i$. We call this game a \emph{reachability timed game}.		
\end{definition} 

Thus, in a timed game, when it is the turn of Player~$i$ to play, if
the play is in location $\ell$, he has to choose a delay $d \in
\mathbb{R}^+$ and a next location $\ell'$ such that $(\ell,a,g,Y,
\ell') \in \Trans$ for some $a \in \Sigma$, $g \in \Guard$ and $Y \in
2^{\Clock}$. If the choice of $d$ respects the guard $g$, then the
choice of Player~$i$ is valid: the clock valuation evolves according
to the past clock valuation, $d$ and $Y$ and location $\ell'$ is
reached. Then, the play continues.
		
\subsection{Regions and region games}
In this section, we consider a bisimulation equivalence on $T_{\Aut}$
(the classical time-abstract bisimulation from~\cite{RajeevAlur.1994})
which allows us to solve the constrained existence in the quotient of
the original timed game (the region game). All along this section we
use the following notations. We denote by $x_i$ the maximum value in
the guard for clock $c_i$. For all, positive number $d \in
\mathbb{R}^+$, $\lfloor d \rfloor$ is the integral part of $d$ and
$\overline{d}$ is fractional part of $d$.

\begin{definition}[$\approx$ and region]
	\label{def:approx}
	\begin{itemize}
		\item Two clock valuations $\nu$ and $\nu'$ are equivalent (written $\nu \approx \nu'$) iff: \emph{(i)} $\lfloor \nu_i \rfloor =  \lfloor \nu'_i \rfloor$ or $\nu_i, \nu'_i > x_i$, for all $i \in \{1, \ldots, k \}$; \emph{(ii)} $\overline{\nu_i} = 0$ iff $\overline{\nu'_i}$, for all $i \in \{1, \ldots, k \}$ with $v_i \leq x_i$ and \emph{(iii)} $\overline{\nu_i} \leq \overline{\nu_j}$ iff $\overline{\nu'_i} \leq \overline{\nu'_j}$ for all $i \neq j \in \{1, \ldots, k\}$ with $\nu_j \leq x_j$ and $\nu_i \leq x_i$.		
			
		\item We extend the equivalence relation to the states ($\approx \subseteq V \times V$) : $(\ell, \nu) \approx (\ell', \nu')$ iff $ \ell = \ell'$ and $\nu \approx \nu'$;
		\item A \emph{region} $r$ is an equivalence class for some $v \in V$: $r = [v]_{\approx} $.
		\end{itemize}
\end{definition}
	
	This equivalence relation on clocks and its extension to states of $T_{\Aut}$ is usual and the following result is well known \cite{RajeevAlur.1994}.
	
\begin{lemma}[\cite{RajeevAlur.1994}]
	Let $(\Aut,\ell_0)$ be a TA, $\approx \subseteq V \times V$ is a bisimulation equivalence on $T_{\Aut}$.
\end{lemma}
		
		It means that if $(\mathcal{G}_{\Aut},v_0)$ is a (reachability) timed game, $\approx$ is a bisimulation equivalence on it. Moreover, it respects the partition. Thus, we can consider the (reachability) quotient game of this game. We call this game the \emph{(reachability) region game}. Notice that $\approx$ respects the target sets, so the reachability quotient game is defined as in Definition~\ref{def:quotientReach}.
		
\begin{definition}[(Reachability) region game]
			Let $(\mathcal{G}_{\Aut},v_0)$ be a (reachability) timed game and $\approx \subseteq V \times V$ be the bisimulation equivalence defined in Definition~\ref{def:approx}, its associated \emph{(reachability) region game} is its associated (reachability) quotient game $(\tilde{\tilde{\mathcal{G}}}_{\Aut},\Class{v_0})$.
\end{definition}
			
We recall~\cite{RajeevAlur.1994} that the size of $\tilde{\tilde{T}}_{\Aut}$, \emph{i.e.,} its number of states (regions) and edges, is in $O((|V|+|\rightarrow|)\cdot 2^{|\delta (\Aut)|})$ where $\delta (\Aut)$ is the binary encoding of the constants (guards and costs) appearing in $\Aut$. Thus $|\tilde{\tilde{T}}_{\Aut}|$ is in $O(2^{|\Aut|})$ where $|\Aut|$ takes into account the locations, edges and constants of $\Aut$. From this follows the following lemma.

\begin{lemma}
	\label{lemma:finiteReach}
	The (reachability) region game $(\tilde{\tilde{\mathcal{G}}}_{\Aut}, \Class{v_0})$ is a finite (reachability) game.
\end{lemma}


Finally, in light of the construction of the reachability region game, the bisimulation equivalence $\approx$ respects the gain functions of the reachability timed game and of the reachability region game.

\begin{lemma}
	Given $(\mathcal{G}_{\Aut},v_0) =  (\timeAr, \allGain, (F_i)_{i\in \Pi})$ be a reachability timed game and $(\tilde{\tilde{\mathcal{G}}}_{\Aut}, \Class{v_0}) = (\tilde{\tilde{\timeAr}}, (\approxGain_i)_{i\in \Pi}, (\tilde{\tilde{F}}_i)_{i \in \Pi})$ its associated region game, $\approx$ respects the gain functions.	
\end{lemma}

\begin{remark}
	Let $\Aut = (\Sigma,\Loc,\Trans,\Clock)$ be a timed automaton, $T_{\Aut} = (\Sigma, V, E)$ be its associated transition system and $\approx$ be the bissimulation equivalence on $T_{\Aut}$ as defined in Definition~\ref{def:approx}, we have that $((\ell,\nu),a,(\ell',\nu')) \in E$ if and only if there exist $g\in \Guard$, $Y \in 2^{\Clock}$ and $d \in \mathbb{R}^+$ such that $(\ell,a,g,Y,\ell') \in \Trans$, $\nu' = [Y\leftarrow 0](v+d)$ and $v+d \models g$. Thus, we abstract the notion of time elapsing in the edeges of the transition system.
	
	Then, since $\approx$ is a bisimulation equivalence on $T_{\Aut}$, for all $((\ell_1,\nu_1),a,(\ell'_1, \nu'_1)) \in E$ and for all $(\ell_2,\nu_2) \in V$ such that $(\ell_1, \nu_1) \approx (\ell_2,\nu_2)$, there exists $(\ell'_2, \nu'_2) \in V$ such that $((\ell_2,\nu_2),a,(\ell'_2,\nu'_2)) \in E$ and $(\ell'_1,\nu'_1) \approx (\ell'_2,\nu'_2)$. The time elapsing between $\nu_1$ and $\nu'_1$ is not necessarly  the same as between $\nu_2$ and $\nu'_2$. Thus, $\approx$ is a timed abstract bisimulation in the classical way~\cite{RajeevAlur.1994}.
\end{remark}
\subsection{Complexity results}

\begin{theorem}	
Given a reachability PPTA $(\Aut, \ell_0)$ and $x,y \in \{0,1\}^{|\Pi|}$, the constrained existence problem in reachability timed games is EXPTIME-complete.
\end{theorem}

The EXPTIME-hardness is due to a reduction from \emph{countdown games} and is inspired by the one provided in~\cite[Section 6.3.3]{brenguier:tel-00827027}. Thus, we only prove the EXPTIME-easiness.

\begin{proof}[EXPTIME-easiness]
  Given a PPTA $(\Aut, \ell_0)$ with target sets $(\Goal_i)_{i\in
    \Pi}$ and given $x,y \in \{0,1\}^{|\Pi|}$.  Thanks to
  Theorem~\ref{thm:eqSPE}, it is equivalent to solve this problem in
  the reachability region game. Moreover, the size of the reachability
  region game is exponential, because its transition system
  $\tilde{\tilde{T}}_{\Aut}$ is exponential in the size of $\Aut$, but
  not in the number of players. Then, since the reachability region
  game is a finite reachability game (Lemma~\ref{lemma:finiteReach}),
  we can apply Theorem~\ref{thm:exptime}. It causes an exponential
  blow-up in the number of players but is polynomial in the size of
  transition system $\tilde{\tilde{T}}_{\Aut}$. Thus, this entire
  procedure runs in (simple) exponential time in the size of the PPTA
  $(\Aut, \ell_0)$ . 
  %
\end{proof}

Notice that, since there always exists an SPE in a finite reachability game~\cite{Ummels06}, there always exists an SPE in the region game and so in the reachability timed game (Theorem~\ref{thm:eqSPE}).


\subsection{Time-bounded reachability, Zenoness and other extensions}
\label{section:thomas}

In this paper, we focus on (qualitative) reachability timed games, and
ignore the effect of Zeno behaviors\footnote{ A run $\rho =(\ell_0,\nu_0) \xrightarrow{d_1,a_1} (\ell_1,\nu_1) \xrightarrow{d_2,a_2}\ldots$ in a timed automaton is said \emph{timed-divergent} if the sequence $(\sum_{j\leq i} d_j)_{i}$ diverges. A timed automaton is \emph{non-Zeno} if any finite run can be extended into a time-divergent run~\cite{DBLP:reference/mc/BouyerFLMO018}.}. Nevertheless we believe that our approach is rather robust and
can be extended to richer objectives and take into account Zeno
behaviors. In the following paragraphs, we try to briefly explain how
this could be achieve.

\paragraph{Time-bounded reachability.}
A natural extension of our framework would be to equip the objective
of each player with a time-bound. Player~$i$ aims at visiting $F_i$
within $TB_i$ time-units. We believe that this time-bound variant of
our constrained problem is decidable. Indeed, for each player, his
time-bound reachability objective can easily be encoded via a
deterministic timed automaton (on finite timed words) $\mathcal{A}_i$.
Given a timed game $\mathcal{G}_b$ equipped with a timed-bounded
objective for each player (described via $\mathcal{A}_i$), we could,
via standart product construction build a new reachability timed game
(without time-bound) $\mathcal{G}$. Solving the constrained existence
problem (with time-bound) in $\mathcal{G}_b$ is equivalent to solving
the constrained existence problem (of Definition~\ref{def:constrained}) in $\mathcal{G}$ (the
constrained being encoded in the $\mathcal{A}_i$'s). This approach
could extend to any property that can be expressed via a deterministic
timed automaton.

\paragraph{Towards $\omega$-regular objectives.}
Let us briefly explain how our approach could be adapted to prove the
decidability of the constrained existence problem for timed games with
$\omega$-regular objectives. For the sake of clarity, we here focus on
parity objective. First, let us notice that the results of
Section~\ref{section:SPE_game_quotient} (including
Theorem~\ref{thm:eqSPE}) apply to a general class of games, including
infinite games with classical $\omega$-regular objectives such as
parity. An algorithm to decide the constrained existence problem
(Definition~\ref{def:constrained}) on parity on finite games can be found in~\cite{Ummels06} via translation into tree automata. Equipped with these two
tools, we believe that we could adapt the definitions and results of
Section~\ref{section:timedGames} to obtain the decidability of the
constrained existence problem for parity timed games. Notice that, in order to obtain our complexity results for finite reachability games, we use other simpler tools than tree automata.

\paragraph{About Zenoness.}
In the present paper, we allow a player to win (or to prevent other
players to win) even if his strategy is responsible of Zeno
behaviors. In~\cite{AlfaroFHMS03}, the authors
propose an elegant approach to \emph{blame} a player that would
prevent divergence of time. The main idea is to transform the
$\omega$-regular objective of each player into another one which will
make him lose if he blocks the time. We believe that this idea could
be exploited in our framework in order to prevent from winning a
``blocking time player''.

%
%
%
 \bibliographystyle{splncs04}
\bibliography{biblio}
\appendix


\section{Proofs of Section~\ref{section:SPE_game_quotient}}

For the sake of clarity, we denote $\Plays_{\tAr}$, $\Hist_{\tAr}$ and $\Hist_{i \tAr}$ by $\widetilde{\Plays}$, $\widetilde{\Hist}$ and $\widetilde{\Hist}_i$ respectively.

\subsection{Proof of Proposition~\ref{prop:uniformSPE}}

In this section, when we consider a history $h = h_0\ldots h_n$ for some $n \in \mathbb{N}$, the length of $h$, denoted by $|h|$, is its number of vertices.\\

This section is devoted to prove Proposition~\ref{prop:uniformSPE}. Let $(\mathcal{G},v_0) = (\Ar, (\Gain_i)_{i\in \pi})$ be a game and $\sim$ be a bisimulation equivalence on $(\mathcal{G},v_0)$ which respects the partition and such that for each $\rho$ and $\rho'$ in $\Plays$, if $\rho \sim \rho'$ then $\Gain(\rho) = \Gain(\rho')$.\\

If there exists an SPE $\tau$ in $(\mathcal{G},v_0)$ which is uniform and such that $\Gain(\outcome{\tau}{v_0})= p$, clearly there exists an SPE $\sigma$ in $(\mathcal{G},v_0)$ such that $\Gain(\outcome{\sigma}{v_0}) = p$. 

The difficult part is the other implication: if there exists an SPE $\sigma$ in $(\mathcal{G},v_0)$ such that $\Gain(\outcome{\sigma}{v_0}) = p$, then there exists an SPE $\tau$ in $(\mathcal{G},v_0)$ which is uniform and such that $\Gain(\outcome{\tau}{v_0})= p.$ Let us prove it. 

Let $\sigma$ be an SPE in $(\mathcal{G},v_0)$ such that $\Gain(\outcome{\sigma}{v_0}) = p$. In order to build $\tau$, we need some additional material and notations that we explain below.
 
\begin{itemize}
	\item for each $h \in \Hist(v_0)$: $[h] = \{ h' \in \Hist(v_0) \mid h \sim h' \}$;
	\item $\mathcal{C}^n= \{ [h] \mid h \in \Hist(v_0) \wedge |h| = n \}$;
	\item $\mathcal{R}: \displaystyle\bigcup_{n \in \mathbb{N}} \mathcal{C}^n \rightarrow \Hist(v_0) \cup \{ \perp \}$ which  allow us to indentify a witness for each class;
	\item $P: \Hist(v_0) \rightarrow \{0,1\}$
\end{itemize}

\subsubsection{Inductive construction of $\mathcal{R}$ and $P$}

The first step, is to choose in a proper way a witness to each each class $[h]$. We proceed by induction on the length of histories. Moreover, we claim that the following properties are satisfied all along the inductive construction.

\fbox{\begin{minipage}{\linewidth-35pt}
\textbf{Invariant 1}: For each $hv \in \Hist(v_0)$ such that $\mathcal{R}([hv]) \neq \perp$, $$hv \sim \mathcal{R}([hv]).$$
\noindent\textbf{Invariant 2:} For each $hv \in \Hist(v_0)$ such that $\mathcal{R}([hv]) \neq \perp$ and $|hv|>1$, $$\mathcal{R}([hv])= \mathcal{R}([h])\Last(\mathcal{R}([hv])).\footnotemark$$
\noindent\textbf{Invariant 3:} 
For each $hv \in \Hist(v_0)$ such that $\mathcal{R}([hv]) \neq \perp$, $$ h'v'\leq \mathcal{R}([hv]) < h'\outcome{\sigma_{\restriction h'}}{v'}$$ for some $h'v'$ such that $P(h'v') = 1$.
\end{minipage}}
\footnotetext{If $h = h_0\ldots h_n$ for some $n \in \mathbb{N}$, $\Last(h)=h_n$.}

Before beginning the induction, we initialize $P$ and $\mathcal{R}$ in the following way: for all $C \in \bigcup_{n \in \mathbb{N}}\mathcal{C}^n$, $\mathcal{R}(C) = \perp$ and for all $h \in \Hist(v_0)$, $P(h)=0$.

\begin{itemize}
	\item For $n =1$ : $\mathcal{C}^1 = \{ [v_0] \}$, we define $P(v_0)=1$. Then , for each $h$ such that $v_0 \leq h < \outcome{\sigma}{v_0}$, we define $\mathcal{R}([h]) =h.$ Thus, Invariant 3 is satisfied with $h'v' = v_0$ and for each $v_0 < hv < \outcome{\sigma}{v_0}$, $\mathcal{R}([hv])$ is defined in this step and satisfies Invariant 2. Since $\mathcal{R}([h]) = h$ for each witness defined in this step, Invariant 1 is satisfied too.
	\item Let us assume that these two invariant are satisfied after step $k$, and let us prove it remains true after step $k+1$.
	\item In this step, we first define $\mathcal{R}$ for each $C \in \mathcal{C}^{k+1}$ such that $\mathcal{R}(C) = \perp$. We know that for all $h_1v_1, h_2v_2 \in C$, $h_1 \sim h_2$ and $\mathcal{R}([h_1]) = \mathcal{R}([h_2])$ are already defined (\emph{i.e.,} $\neq \perp$). Moreover, by Invariant 1, $h_1 \sim \mathcal{R}([h_1])$, let $h = \mathcal{R}([h_1])$, by bisimulation $\sim$, there exists $v \in V$ such that $h_1v_1 \sim hv$. We define $P(hv) = 1$ and $\mathcal{R}(C) = hv$. Then, for all $h_2v_2 \in C$, $h_2v_2 \sim h_1v_1 \sim hv$ this implies that $h_2v_2 \sim \mathcal{R}([h_2v_2])$ (Inv 1 ok). For all $h_2v_2 \in C$, $\mathcal{R}([h_2v_2]) = \mathcal{R}(C) = hv = \mathcal{R}([h_1])v = \mathcal{R}([h_2])v =  \mathcal{R}([h_2]) \Last(\mathcal{R}([h_2v_2]))$ (Inv 2 ok). Moreover, for all $h_2v_2 \in C$, $\mathcal{R}([h_2v_2]) = hv$ and $P(hv)=1$, since $hv \leq hv < h \outcome{\sigma_{\restriction h}}{v}$ (Inv 3 ok).\\
	
	Now, we extend the construction of $\mathcal{R}$ and $P$ from $hv$ in the following way: \\
	\fbox{ $\forall h'v' \in \Hist(v_0)$ such that  $hv < h'v' < h \outcome{\sigma_{\restriction h}}{v}$ we define $\mathcal{R}([h'v'])=h'v'$. $\star$}
	
	Now, we have to prove that the invariants remains satisfied for all these new defined classes.
	\begin{itemize}
		\item $\forall \overline{h}\overline{v} \in [h'v']$: $\overline{h}\overline{v} \sim h'v' = \mathcal{R}([h'v']) = \mathcal{R}([\overline{h}\overline{v}])$ (Inv 1 ok);
		\item $\forall \overline{h}\overline{v} \in [h'v']$, we have that $\overline{h}\overline{v} \sim h'v'$ thus: $\mathcal{R}([\overline{h}\overline{v}]) = \mathcal{R}([h'v']) = \mathcal{R}([h']) \Last(\mathcal{R}([h'v']))$ (by construction $\star$). Thus, since $\overline{h} \sim h'$: $\mathcal{R}([h'])= \mathcal{R}([\overline{h}])$ (Inv 2 ok).
		
		\item $\forall \overline{h}\overline{v} \in [h'v']$, we have by construction $\star$ that $P(hv) = 1 $ and $hv < h'v' < h \outcome{\sigma_{\restriction h}}{v}$. Since $h'v' = \mathcal{R}([h'v'])$ and $\mathcal{R}([h'v']) = \mathcal{R}([\overline{h}\overline{v}])$ ($h'v' \sim \overline{h}\overline{v}$), we are done (Inv 3 ok).
	\end{itemize}	
\end{itemize}

\subsubsection{Construction of $\tau$}

To build the uniform strategy profile $\tau$, we proceed as follows: for all $n \in \mathbb{N}$, for all $C \in \mathcal{C}^n$, for all $h \in C$, by assuming that $\Last(h) \in V_i$:

\begin{itemize}
	\item If $\mathcal{R}([h]) = h$ ($h$ is a witness, thus we want to follow $\sigma$): $\tau_i(h) = \sigma_i(h)$;
	\item If $\mathcal{R}([h]) \neq h$ (we simulate $\sigma$): we know by Invariant 1 that $h \sim \mathcal{R}([h])$, thus in particular $\Last(h) \sim \Last(\mathcal{R}([h]))$, by bisimulation $\sim$, there exists $x\in V$ such that $\Last(h)x \sim \Last(\mathcal{R}([h]))\sigma_i(\mathcal{R}([h]))$. Thus, we define $\tau_i(h) = x$.
\end{itemize}

We state now, some properties about $\tau$ and $\sigma$. First, we define $\Wit = \{ h \in \Hist(v_0) \mid \exists C \in \bigcup_{n\in\mathbb{N}}\mathcal{C}^n \text{ st. } \mathcal{R}(C)= h \}.$

\begin{lemma}
	\label{lemma:egualiteTauSigma}
	For all $h \in \Hist(v_0)$ such that $h \in \Wit$ and $\Last(h) \in V_i$: $\tau_i(h) = \sigma_i(h)$.
\end{lemma}
\begin{proof}
	This assertion is true due to the construction of $\tau$.
\end{proof}

\begin{lemma}
	\label{lemma:tauUniform}
	For all $h,h' \in \Hist(v_0)$ such that $h \sim h'$: $\tau_i(h) \sim \tau_i(h')$ by assuming that $\Last(h) \in V_i$.
\end{lemma}

Notice that,  since $\sim$ respects the partition, if $\Last(h) \in V_i$ then $\Last(h')\in V_i$, and vice versa.

\begin{proof}
	Let $h,h' \in \Hist(v_0)$ such that $h \sim h'$ and $\Last(h) \in V_i$ for some $i \in \Pi$ then $\Last(h') \in V_i$. We have that $\mathcal{R}([h]) = \mathcal{R}([h'])$. By construction, $\tau_i(h) \sim \sigma_i(\mathcal{R}([h]))$ and $\tau_i(h') \sim \sigma_i(\mathcal{R}([h']))$, by transitivity, we have: $\tau_i(h) \sim \tau_i(h')$. \qed
\end{proof}

\begin{lemma}
	\label{lemma:witness}
	For all $h \in \Wit$, $h\tau_i(h) \in \Wit$ (by assuming that $\Last(h) \in V_i$ for some $i \in \Pi$).
\end{lemma}

\begin{proof}	
	Let $h \in \Wit$, such that $\Last(h) \in V_i$ for some $i \in \Pi$. Since $h \in \Wit$, by Invariant 3, there exists $h'v' \in \Hist(v_0)$ such that:
	$$ h'v' \leq h < h' \outcome{\sigma_{\restriction h'}}{v'}.$$ Thus, we have that $$ h'v' \leq h\sigma_i(h) < h' \outcome{\sigma_{\restriction h'}}{v'}.$$
	It follows by construction of $\mathcal{R}$, that $h\sigma_i(h) \in \Wit$. Moreover, $h \in \Wit$ implies that $\tau_i(h) = \sigma_i(h)$ (by Lemma~\ref{lemma:egualiteTauSigma}). Thus, $h\tau_i(h) \in \Wit$. \qed
\end{proof}

\begin{lemma}
	\label{lemma:outcomeWit}
	For all $hv \in \Wit$, $\outcome{\sigma_{\restriction h}}{v} = \outcome{\tau_{\restriction h}}{v}$.
\end{lemma}

\begin{proof}
	Let $hv \in \Wit$, let $\rho = \outcome{\sigma_{\restriction h}}{v}$ and let $\overline{\rho} = \outcome{\tau_{\restriction h}}{v}$.  Let us prove by induction that for all $n \in \mathbb{N}$: \begin{enumerate} \item $\rho_n = \overline{\rho}_n$; \item $h\rho_0\ldots \rho_n \in \Wit$.\end{enumerate}
	
	For $n = 0$, $\rho_0 = v$ and $\overline{\rho}_0=v$. And by hypothesis, $hv \in \Wit$. Let us assume that both assertions are satisfied for all $n $ such that $n \leq k$. Let us prove that it remains true for $n = k+1$. By assuming that $\overline{\rho}_k \in V_i$,
\begin{enumerate}
	\item\begin{align*}
		\overline{\rho}_{k+1} &= \tau_i(h\overline{\rho}_0\ldots \overline{\rho}_k) \\
							  &= \tau_i(h\rho_0\ldots\rho_k) & \text{ By IH, } \overline{\rho}_0\ldots \overline{\rho}_k = \rho_0 \ldots \rho_k\\
							  &= \sigma_i(h\rho_0\ldots\rho_k)& \text{ By IH, } h\rho_0\ldots \rho_k \in \Wit \text{ and by Lemma~\ref{lemma:egualiteTauSigma}}\\
							  &= \rho_{k+1}.
	\end{align*}
	\item By IH, $h\rho_0\ldots\rho_k \in \Wit$, moreover we have that, by Lemma~\ref{lemma:egualiteTauSigma}:
	
		 $$h\rho_0\ldots\rho_k\rho_{k+1} = h\rho_0\ldots\rho_k \sigma_i(h\rho_0\ldots\rho_k)
		 							   = h\rho_0\ldots\rho_k \tau_i(h\rho_0\ldots\rho_k) $$
									   
	And by Lemma~\ref{lemma:witness}, we can conclude that $h\rho_0\ldots\rho_k \tau_i(h\rho_0\ldots\rho_k) \in \Wit$. \qed
\end{enumerate}
\end{proof}

\subsubsection{Proof that $\tau$ is an uniform SPE with gain profile $p$}
There is still to prove that $\tau$ is an uniform SPE in $(\mathcal{G},v_0)$ such that $\Gain(\outcome{\tau}{v_0}) = p$. By Lemma~\ref{lemma:tauUniform}, $\tau$ is uniform, let us prove this is an SPE with the gain profile $p$.

\begin{proof}
	First, since $v_0 \in \Wit$ and by Lemma~\ref{lemma:outcomeWit}, we have that $\outcome{\sigma}{v_0} = \outcome{\tau}{v_0}$. Thus, in particular, $\Gain(\outcome{\tau}{v_0}) = \Gain(\outcome{\sigma}{v_0})=p$.\\
	
	By absurdum, let us assume that $\tau$ is not an SPE in $(\mathcal{G},v_0)$. It means that there exist $hv \in \Hist(v_0)$, $i \in \Pi$ and a strategy $\tau'_i$ of Player~$i$ in $(\mathcal{G}_{\restriction h},v)$ such that $\tau'_i$ is a profitable deviation of $\tau_{i\restriction h}$, \emph{i.e.,}
	
	\begin{equation} \Gain_i(h \outcome{\tau_{\restriction h}}{v}) < \Gain_i(h \outcome{\tau'_i, \tau_{-i\restriction h}}{v}).\label{eq:proof0}\end{equation}
	
	Let $h'v' = \mathcal{R}([hv]) = \mathcal{R}([h])\Last(\mathcal{R}([hv]))$ by Invariant 2.\\
	
	\noindent\underline{First step:} Let $\rho = \outcome{\tau_{\restriction h}}{v}$ and $\rho' = \outcome{\tau_{\restriction h'}}{v'}$, let us prove by induction that for all $n \in \mathbb{N}$: \begin{enumerate}
	\item $\rho_n \sim \rho'_n$;
	\item $h'\rho'_0 \ldots \rho'_n \in \Wit$.
	\end{enumerate}
	
	For $n=0$, we have that $\rho_0=v$ and $\rho'_0= v'$, thus $v \sim v'$ since $hv \sim h'v'$. Moreover, $h'v' \in \Wit$ by hypothesis. Let us assume that these two properties are satisfied for all $n$ such that $n \leq k$, let us prove they remain true for $n = k+1$. Let us assume that $\rho_k \in V_j$ for some $j\in \Pi$, since $\sim$ respects the partition and due to the fact that $\rho_k \sim \rho'_k$ by IH, we have that $\rho'_k \in V_j$.
	
	\begin{enumerate}
		\item \begin{align*}
		\rho_{k+1} &= \tau_j(h\rho_0 \ldots \rho_k)\\
				   &\sim \tau_j(h'\rho'_0\ldots \rho'_k) & \text{By IH, } h\rho_0\ldots\rho_k \sim h'\rho'_0\ldots \rho'_k \text{ and by Lemma~\ref{lemma:tauUniform}}\\
				   &= \rho'_{k+1}
		\end{align*}
		\item $h'\rho'_0\ldots\rho'_k\rho'_{k+1} = h'\rho'_0\ldots\rho'_k \tau_i(h'\rho'_0\ldots\rho'_k)$, $h'\rho'_0\ldots\rho'_k \in \Wit$ by IH, thus by Lemma~\ref{lemma:witness}: $h'\rho'_0\ldots\rho'_k \tau_i(h'\rho'_0\ldots\rho'_k) \in \Wit$.
	\end{enumerate}
	It allows us to state by (1) that $h\rho \sim h'\rho'$, thus by hypothesis on $\sim$, we have that \begin{equation}\Gain(h\outcome{\tau_{\restriction h}}{v}) = \Gain(h'\outcome{\tau_{\restriction h'}}{v'}).\label{eq:proof1}\end{equation}
	
	By (2) and Lemma~\ref{lemma:egualiteTauSigma}, we have that $\outcome{\tau_{\restriction h'}}{v'} = \outcome{\sigma_{\restriction h'}}{v'}$ and thus:
	
	\begin{equation} \Gain(h'\outcome{\tau_{\restriction h'}}{v'}) = \Gain(h'\outcome{\sigma_{\restriction h'}}{v'}).\label{eq:proof2}\end{equation}
		
\noindent\underline{Second step:}	Let $\rho = \outcome{\tau'_i, \tau_{-i\restriction h}}{v}$, we will build a strategy $\tilde{\tau}_i$ in $(\mathcal{G}_{\restriction h'}, v')$ such that $\rho \sim \outcome{\tilde{\tau}_i, \tau_{-i\restriction h'}}{v'}$. Let $\tilde{p} \in \Hist_i(v')$ and let us assume that $\tilde{p} = \tilde{p}_0 \ldots \tilde{p}_m$ for some $m \in \mathbb{N}$.
\begin{itemize}
	\item If $\tilde{p}\sim \rho_0\ldots \rho_m$, then $\rho_m \in V_i$ ($\sim$ respects the partition) and $\rho_{m+1} = \tau'_i(\rho_0\ldots \rho_m)$. Thus, by $\sim$ there exists $x \in V$ such that $\tilde{p}x \sim \rho_0\ldots \rho_m\rho_{m+1}$. We define $\tilde{\tau}_i(\tilde{p}) = x$. Thus, $\tilde{\tau}_i(\tilde{p}_0 \ldots \tilde{\rho}_m) \sim \tau'_i(\rho_0\ldots \rho_m)$.
	\item Otherwise, we define $\tilde{\tau}_i(\tilde{p}) = x$ for some $x \in \Succ(\tilde{p}_m)$.
\end{itemize} 

Let $\tilde{\rho} = \outcome{\tilde{\tau}_i, \tau_{-i \restriction h'}}{v'}$, let us prove that for all $n \in \mathbb{N}$, $\tilde{\rho}_n \sim \rho_n$.
For $n=0$, $\tilde{\rho}_0 = v'$ and $\rho_0 = v$, since $h'v'\sim hv$, $v'\sim v$. Let us assume that this property is true for all $n\leq k$ and let us prove it remains true for $n=k+1.$
\begin{itemize}
	\item If $\tilde{\rho}_k \in V_i$, then since $\tilde{\rho}_k \sim \rho_k$ by IH, $\rho_k \in V_i$ ($\sim$ respects the partition). It follows:
	\begin{align*}
		\tilde{\rho}_{k+1} &= \tilde{\tau}_i(\tilde{\rho}_0\ldots \tilde{\rho}_k)\\
						   &\sim \tau'_i(\rho_0\ldots \rho_k) & \text{ By IH, }\tilde{\rho}_0\ldots \tilde{\rho}_k \sim \rho_0 \ldots \rho_k\text{ and by construction of }\tilde{\tau}_i\\
						   &= \rho_{k+1}.
	\end{align*}
	\item If $\tilde{\rho}_k \in V_j$ ($j \neq i$), as previously $\rho_k \in V_j$. Thus:
	\begin{align*}
		\tilde{\rho}_{k+1} &= \tau_{j \restriction h'}(\tilde{\rho}_0\ldots \tilde{\rho}_k) = \tau_{j}(h'\tilde{\rho}_0\ldots \tilde{\rho}_k)\\
						   &\sim \tau_j(h\rho_0\ldots\rho_k) \quad \text{ By IH, } h'\tilde{\rho}_0\ldots \tilde{\rho}_k \sim h\rho_0 \ldots \rho_k \text{ and by Lemma~\ref{lemma:tauUniform}}\\
						   &= \rho_{k+1}.
		\end{align*}
\end{itemize}

From this we have that $h\rho \sim h'\tilde{\rho}$ and in particular:

\begin{equation}\Gain_i(h\outcome{\tau'_i, \tau_{-i\restriction h}}{v}) = \Gain_i(h'\outcome{\tilde{\tau}_i, \tau_{-i \restriction h'}}{v'}).\label{eq:proof3}\end{equation}
	
\noindent\underline{Third step:} From $\tilde{\tau}_i$, we build $\tilde{\sigma}_i$ in $(\mathcal{G}_{\restriction h'},v')$ which is a profitable devition of $\sigma_{i \restriction h'}$. Let $p \in \Hist_i(v')$

\begin{itemize}
	\item If $h'p \in \Wit$, we consider $\mathcal{R}([h'p\tilde{\tau}_i(p)]) = \mathcal{R}([h'p]) \Last(\mathcal{R}([h'p\tilde{\tau}_i(p)])$ by Invariant 2.
	Let $x = \Last(\mathcal{R}([h'p\tilde{\tau}_i(p)]) $. We definie $\tilde{\sigma}_i(p) = x$, in particular, we have that $\mathcal{R}([h'p\tilde{\tau}_i(p)]) = \mathcal{R}([h'p]) \tilde{\sigma}_i(p)$ and thus $\mathcal{R}([h'p]) \tilde{\sigma}_i(p) \in \Wit$.
	\item If $h'p \not\in \Wit$, we define $\tilde{\sigma}_i(p) = \tilde{\tau}_i(p)$.
\end{itemize}
	
	Let $\pi = \outcome{\tilde{\sigma}_i, \sigma_{-i \restriction h'}}{v'}$ and $\pi' = \outcome{\tilde{\sigma}_i, \tau_{-i \restriction h'}}{v'}$.
	Let us prove that for all $n \in \mathbb{N}$:
	
	\begin{enumerate} 
		\item $\pi_n = \pi'_n$;
		\item $h'\pi_0\ldots \pi_n \in \Wit$
	\end{enumerate}
	
	For $n = 0$, we have that $\pi_0 = v' = \pi'_0$. Moreover, $h'v' \in \Wit$ by hypothesis. Let us assume that these two properties are true for all $n \leq k$ and let us prove that they remain true for $n = k+1$.
	
	\begin{itemize}
		\item If $\pi_k\in V_i$, then by IH, $\pi_k =\pi'_k \in V_i$. 
		\begin{enumerate}
		\item \begin{align*}
			\pi_{k+1} &= \tilde{\sigma}_i(\pi_0\ldots \pi_k)\\
					  &= \tilde{\sigma}_i(\pi'_0\ldots \pi'_k) & \text{By IH, } \pi_0\ldots\pi_k = \pi'_0\ldots\pi'_k.\\
					  &=\pi'_{k+1}
			\end{align*}
		\item \begin{align*}h'\pi_0\ldots\pi_k\pi_{k+1} &=h'\pi_0\ldots\pi_k \tilde{\sigma}_i(\pi_0\ldots\pi_k)\\
								&= \mathcal{R}([h'\pi_0\ldots\pi_k]) \tilde{\sigma}_i(\pi_0\ldots\pi_k) & \text{By IH, } h'\pi_0\ldots \pi_k \in \Wit \\
								&\in \Wit & \text{ By construction of } \tilde{\sigma}_i.\end{align*} 
		\end{enumerate}
		
		\item If $\pi_k\in V_j$ ($j \neq i$), then by IH, $\pi_k =\pi'_k \in V_j$.
		\begin{enumerate}
			\item \begin{align*}
			\pi_{k+1} &= \sigma_j(h'\pi_0\ldots \pi_k)\\
		  &= \tau_j(h'\pi_0\ldots \pi_k) & \text{ By IH, } h'\pi_0\ldots \pi_k \in \Wit \text{ and by Lemma~\ref{lemma:egualiteTauSigma}}\\
		  &= \tau_j(h'\pi'_0\ldots \pi'_k) & \text{ By IH}\\
		  &= \pi'_{k+1}.
		  \end{align*}
		  	\item 
			\begin{align*}
				h'\pi_0\ldots\pi_k\pi_{k+1} &= h'\pi_0\ldots\pi_k\sigma_j(h'\pi_0\ldots\pi_k) \\
											&= h'\pi_0\ldots\pi_k \tau_j(h'\pi_0\ldots\pi_k) & \text{ By IH, }h'\pi_0\ldots\pi_k\in \Wit \text{ and by Lemma~\ref{lemma:egualiteTauSigma} }\\
											&\in \Wit &\text{ By Lemma~\ref{lemma:witness}}
			\end{align*}
		\end{enumerate}
		\end{itemize}
		
Thus, we can conclude that:

\begin{equation}\Gain_i(h'\outcome{\tilde{\sigma}_i, \sigma_{-i \restriction h'}}{v'}) = \Gain_i(h'\outcome{\tilde{\sigma}_i, \tau_{-i \restriction h'}}{v'}). \label{eq:proof4}\end{equation}
	
Now, we want to prove that $\pi' = \outcome{\tilde{\sigma}_i, \tau_{-i \restriction h'}}{v'} \sim \tilde{\rho} = \outcome{\tilde{\tau}_i, \tau_{-i \restriction h'}}{v'}$. Let us recall, that from the second step, we know that $\tilde{\rho} \sim \rho = \outcome{\tau'_i, \tau_{-i\restriction h}}{v}$.
Let us prove that for all $n \in \mathbb{N}$: $\pi'_n \sim \tilde{\rho}_n$.

For $n = 0$: $\pi'_0 = v' = \tilde{\rho}_0$. Let us assume that this property is true for all $n \leq k$ and let tus prove that it remains true for $n = k+1$.
\begin{itemize}
	\item \textbf{If $\pi'_k \in V_i$} then, by IH we have that $\pi'_k \sim \tilde{\rho}_k$ and so $\tilde{\rho}_k \in V_i$. 
	\begin{align*}
		\pi'_{k+1} &= \tilde{\sigma}_i(\pi'_0\ldots \pi'_{k})\\
				   &= \Last(\mathcal{R}([h'\pi'_0\ldots\pi'_k \tilde{\tau}_i(\pi'_0\ldots \pi'_k)])) & h'\pi'_0\ldots\pi'_k \in \Wit\\
				   &\sim \tilde{\tau}_i(\pi'_0\ldots \pi'_k)
		\end{align*}
		By IH, we know that $\pi'_0\ldots \pi'_k\sim \tilde{\rho}_0 \ldots \tilde{\rho}_k$ and by hypothesis, we have that $\tilde{\rho}_0 \ldots \tilde{\rho}_k \sim \rho_0\ldots\rho_k$. It follows from the construction of $\tilde{\tau}_i$ that $\tilde{\tau}_i(\pi'_0\ldots \pi'_k) \sim \tau'_i(\rho_0\ldots\rho_k)$ and $\tilde{\tau}_i(\tilde{\rho}_0 \ldots \tilde{\rho}_k) \sim \tau'_i(\rho_0\ldots\rho_k)$. Thus, by transitivity, $\pi'_{k+1} \sim \tilde{\tau}_i(\tilde{\rho}_0 \ldots \tilde{\rho}_k) = \tilde{\rho}_{k+1}$.
		\item \textbf{If $\pi'_k \in V_j$ ($j \neq i$)} then as previously $\tilde{\rho}_k \in V_j$.
		\begin{align*}
			\pi'_{k+1} &= \tau_j(h'\pi'_0\ldots \pi'_k)\\
					   &\sim \tau_j(h'\tilde{\rho}_0\ldots \tilde{\rho}_k) & \text{ By IH, } h'\pi'_0\ldots \pi'_k \sim h'\tilde{\rho}_0\ldots \tilde{\rho}_k \text{ and by Lemma~\ref{lemma:tauUniform}}\\
					   &= \tilde{\rho}_{k+1}.
		\end{align*}
\end{itemize}

Thus $h'\pi' \sim h'\tilde{\rho}$ and it follows that:

\begin{equation}\Gain_i(h'\outcome{\tilde{\sigma}_i, \tau_{-i \restriction h'}}{v'}) = \Gain_i(h'\outcome{\tilde{\tau}_i, \tau_{-i \restriction h'}}{v'}). \label{eq:proof5}\end{equation}
	
\underline{Fourth step: putting all together}

By~\eqref{eq:proof0},\eqref{eq:proof1},\eqref{eq:proof2},\eqref{eq:proof3},\eqref{eq:proof4} and \eqref{eq:proof5}, we can conclude that 

$$ \Gain_i(h'\outcome{\sigma_{\restriction h'}}{v'}) < \Gain_i(h'\outcome{\tilde{\sigma}_i, \sigma_{-i\restriction h'}}{v'}).$$

Thus, there exists a profitable deviation of $\sigma_{i\restriction h'}$ for Player $i$ in $(\mathcal{G}_{\restriction h'},v')$. This is impossible, since $\sigma$ is an SPE in $(\mathcal{G},v_0)$.\qed

\end{proof}


\subsection{Proof of Theorem~\ref{thm:eqSPE}}

In this section we prove Theorem~\ref{thm:eqSPE}. In order to do so, we prove the two implications of the equivalence in two different propositions. 

\begin{proposition}
	Let $(\mathcal{G},v_0) = (\Ar, (\Gain_i)_{i\in \pi})$ be a game and $\initQ = ( \tAr, (\tilde{\Gain}_i)_{i \in \Pi})$ its associated quotient game where $\sim$ is a bisimulation equivalence on $(\mathcal{G},v_0)$. If $\sim$ respects the partition and the gain functions, we have that: if there exists an SPE $\sigma$ in  $(\mathcal{G},v_0)$ such that $\Gain(\outcome{\sigma}{v_0}) = p $ for some $p \in \{0,1\}^{|\Pi|}$ then there exists an SPE $\tau$ in $\initQ$ such that $\tilde{\Gain}(\outcome{\tau}{\Class{v_0}}) = p$.	
\end{proposition}
\begin{proof}
	Let $(\mathcal{G},v_0) = (\Ar, (\Gain_i)_{i\in \pi})$ be a game and $\initQ = ( \tAr, (\tilde{\Gain}_i)_{i \in \Pi})$ its associated quotient game where $\sim$ is a bisimulation equivalence on $(\mathcal{G},v_0)$ which respects the partition and the gain functions. We assume that there exists an SPE $\sigma$ in $(\mathcal{G},v_0)$ such that $\Gain(\outcome{\sigma}{v_0}) = p $ for some $p \in \{0,1\}$. 
	
Without loss of generality, we can assume thanks to Proposition~\ref{prop:uniformSPE} that $\sigma$ is uniform, \emph{i.e.,} for all histories $h$,$h'\in \Hist(v_0)$	such that $\Last(h) \in V_i \Leftrightarrow \Last(h') \in V_i$, $\sigma_i(h) \sim \sigma_i(h')$.\\

Let $\tilde{h} \in \widetilde{\Hist}([v_0])$ be a history in the quotient game, by bisimulation $\sim_q \subseteq V \times \tilde{V}$, there exists $h = h_0\ldots h_n \in \Hist(v_0)$ such that $h \sim_q \tilde{h} = [h_0]\ldots[h_n]$. Let $v \in V$ be the vertex such that $\sigma_i(h) = v$, by assuming that $\Last(h)\in V_i$. We have that $v \sim_q [v]$ and, we define $\tau_i(\tilde{h}) = [v]$.

\fbox{\begin{minipage}{\linewidth-35pt}
\textbf{CLAIM 1:} $\forall \tilde{h} \in \widetilde{\Hist}([v_0])$, $\forall h \in \Hist(v_0)$ such that $h \sim_q \tilde{h}$, if $\Last(h)\in V_i$, $\sigma_i(h) \sim_q \tau_i(\tilde{h})$.\\
\underline{Proof:} Let $\tilde{h}\in \widetilde{\Hist}([v_0])$ and $h \in \Hist(v_0)$ such that $\Last(h)\in V_i$ for some $i\in \Pi$ and $h \sim_q \tilde{h}$. By construction of $\tau$, there exists $h'\in \Hist_i(v_0)$ such that $h' \sim_q \tilde{h}$ and $\tau_i(\tilde{h}) = [\sigma_i(h')]$. If $h \sim_q \tilde{h}$ and $h'\sim_q \tilde{h}$, we have that $h \sim h'$. Thus, by uniformity of $\sigma$, $\sigma_i(h)\sim \sigma_i(h')$ and in particular $[\sigma_i(h)] = [\sigma_i(h')]$. In conclusion, $\sigma_i(h) \sim_q [\sigma_i(h)] = [\sigma_i(h')] = \tau_i(\tilde{h})$. \qed
\end{minipage}}

Let $\rho = \outcome{\sigma}{v_0}$ and $\tilde{\rho} = \outcome{\tau}{\Class{v_0}}$. Let us prove that: $ \forall n \in \mathbb{N}$ $\rho_n \sim_q \tilde{\rho}_n$. Thus, $\rho \sim_q \tilde{\rho}$ and since $\sim$ respects the gain functions, $\Gain(\rho) = \tilde{\Gain}(\tilde{\rho})=p$.

For $n = 0$: $\rho_0= v_0 \sim_q [v_0] = \tilde{\rho}_0$. We assume that this is true for all $n \leq k$ and we prove it remains true for $n = k+1$. By induction hypothesis, we have that $\rho_0\ldots\rho_k \sim_q \tilde{\rho}_0 \ldots \tilde{\rho}_k$. By $\sim_q$, $\rho_k \in V_i$ if and only if $\tilde{\rho}_k \in \tilde{V}_i$. Thus, $\rho_{k+1} = \sigma_i(\rho_0\ldots \rho_k) \sim_q \tau_i(\tilde{\rho}_0 \ldots \tilde{\rho}_k) = \tilde{\rho}_{k+1}$ by Claim 1. \\

To conclude, we have to prove that $\tau$ is an SPE in $\initQ$. Ad absurdum, we assume that there exists $\tilde{h}\tilde{v}\in \widetilde{\Hist}(\Class{v_0})$ 
such that there exists a player $i\in \Pi$ and a profitable deviation $\tau'_i$ of $\tau_{i\restriction \tilde{h}}$ in $(\tilde{\mathcal{G}}_{\restriction \tilde{h}}, \tilde{v})$,
 \emph{i.e.,} \begin{equation} \tGain_i(\tilde{h}\outcome{\tau_{\restriction \tilde{h}}}{\tilde{v}}) <  \tGain_i(\tilde{h}\outcome{\tau'_i, \tau_{-i \restriction \tilde{h}}}{\tilde{v}}). \label{eq:proof11}\end{equation}
 
 By bisimulation $\sim_q$, there exists $hv \in \Hist(v_0)$ such that $hv \sim_q \tilde{h}\tilde{v}$. We prove that Player~$i$ has a profitable deviation of $\sigma_{i\restriction h}$ in $(\mathcal{G}_{\restriction h},v)$. From which a contradiction follows since $\sigma$ has to be an SPE in $\initG$.\\
 
 We build the profitable deviation $\sigma'_i$. Let $p \in \Hist_i(v)$ a history such that $p=vp_1p_2\ldots p_m$ for some $m \in \mathbb{N}$. By bisimulation $\sim_q$, there exists a unique $\tilde{p}$ such that $p = \Class{v}\Class{p_1}\ldots\Class{p_m}$ and thus $p \sim_q \tilde{p}$. Let $r \in \tilde{V}$ be such that $\tau'_i(\tilde{p}) = r$. By bisimulation $\sim_q$, there exists $x \in V$ such that $\rho_m x \sim_q \Class{\rho_m}r$. We define $\sigma'_i(p) = x$.  In particular, $\sigma'_i(p) \sim_q \tau'_i(\tilde{p})$.\\
 
 Let $\rho = \outcome{\sigma'_i, \sigma_{-i \restriction h}}{v} = v \rho_1 \rho_2 \ldots$ and $\tilde{\rho}= \outcome{\tau'_i, \tau_{-i \restriction \tilde{h}}}{\tilde{v}}= \tilde{v}\tilde{\rho}_1 \tilde{\rho}_2 \ldots$. Let us show by induction that for all $n$, $\rho_n \sim_q \tilde{\rho}_n$. It means that $h\rho \sim_q \tilde{h}\tilde{\rho}$ and since $\sim_q$ respects the gain functions, 
 
 \begin{equation}  \Gain_i(h\outcome{\sigma'_i, \sigma_{-i \restriction h}}{v}) = \Gain_i(h\rho) =  \tGain_i(\tilde{h}\tilde{\rho}) = \tGain_i(\tilde{h}\outcome{\tau'_i, \tau_{-i \restriction \tilde{h}}}{\tilde{v}}). \label{eq:proof12}\end{equation}
 
 For $n=0$: $\rho_0= v \sim_q \tilde{v} = \tilde{\rho}_0$. Assume that this property is true for all $n \leq k$ and let us prove it remains true for $n = k+1$.
 
 \begin{itemize}
	 \item \textbf{First case: if $\rho_k \in V_i$}, by IH $\rho_k \sim_q \tilde{\rho}_k$ and thus $\tilde{\rho}_k \in \tilde{V}_i$. It follows that:
	 \begin{align*} \rho_{k+1} &= \sigma'_i(\rho_0 \ldots \rho_k) \\
		 					   &\sim_q \tau'_i(\tilde{\rho}_0\ldots \tilde{\rho}_k)&\text{ by construction of } \sigma_i \text{ and } \rho_0\ldots\rho_k \sim_q \tilde{\rho}_0 \ldots \tilde{\rho}_k \text{ (IH)}\\
							   &= \tilde{\rho}_{k+1}
	\end{align*}
	
	\item \textbf{Seconde case: if $\rho_k \in V_j$ with $(j \neq i)$} then as previously $\tilde{\rho}_k \in \tilde{V}_j$ and we have:
	\begin{align*}
		\rho_{k+1} &= \sigma_{j \restriction h}(\rho_0 \ldots \rho_k)\\
		           &\sim_q \tau_{j \restriction \tilde{h}}(\tilde{\rho}_0 \ldots \tilde{\rho}_k) & \rho_0 \ldots \rho_k \sim_q \tilde{\rho}_0\ldots \tilde{\rho}_k \text{(HI) and Claim 1.}\\
				   &= \tilde{\rho}_{k+1}
	\end{align*}
\end{itemize}

	There is still to prove that \begin{equation} \tGain_i(\tilde{h}\outcome{\tau_{\restriction \tilde{h}}}{\tilde{v}}) = \Gain_i(h\outcome{\sigma_{\restriction h}}{v}). \label{eq:proof13} \end{equation}
	
	By Claim~1, we have that  $ \outcome{\sigma_{\restriction h}}{v} \sim_q \outcome{\tau_{\restriction \tilde{h}}}{\tilde{v}}$ thus $ h\outcome{\sigma_{\restriction h}}{v} \sim_q \tilde{h} \outcome{\tau_{\restriction \tilde{h}}}{\tilde{v}}$. The fact that $\sim_q$ respects the gain functions concludes the reasonment.\\

 By \eqref{eq:proof11},\eqref{eq:proof12} and \eqref{eq:proof13}, we conclude that $\sigma'_i$ is a profitable deviation in $(\mathcal{G}_{\restriction h}, v)$. \qed
 
\end{proof}

\begin{proposition}
	Let $(\mathcal{G},v_0) = (\Ar, (\Gain_i)_{i\in \pi})$ be a game and $\initQ = ( \tAr, (\tilde{\Gain}_i)_{i \in \Pi})$ its associated quotient game where $\sim$ is a bisimulation equivalence on $(\mathcal{G},v_0)$. If $\sim$ respects the partition and the gain functions, we have that: if there exists an SPE $\tau$ in $\initQ$ such that $\tilde{\Gain}(\outcome{\tau}{\Class{v_0}}) = p$  for some $p \in \{0,1\}^{|\Pi|}$ then there exists an SPE $\sigma$ in  $(\mathcal{G},v_0)$ such that $\Gain(\outcome{\sigma}{v_0}) = p $.	
\end{proposition}

\begin{proof}
	Let $(\mathcal{G},v_0) = (\Ar, (\Gain_i)_{i\in \pi})$ be a game and $\initQ = ( \tAr, (\tilde{\Gain}_i)_{i \in \Pi})$ its associated quotient game where $\sim$ is a bisimulation equivalence on $(\mathcal{G},v_0)$ which respects the partition and the gain functions. We assume that there exists an SPE $\tau$ in $\initQ$ such that $\tGain(\outcome{\tau}{\Class{v_0}}) = p $ for some $p \in \{0,1\}$. \\
	
	Let $h \in \Hist(v_0)$ such that $\Last(h) \in V_i$ for some $i \in \Pi$. Thanks to bisimulation $\sim_q$, there exists a unique $\tilde{h} \in \widetilde{\Hist}_i(\Class{v_0})$ such that $h \sim_q \tilde{h}$ $(\star)$. We have that $\tau_i(\tilde{h}) = \tilde{v}$ for some $\tilde{v}\in \tilde{V}$, thus by $\sim_q$ there exists $v \in V$ such that $hv \sim_q \tilde{h}\tilde{v}$. We define $\sigma_i(h) = v$.

	\fbox{\begin{minipage}{\linewidth-35pt}
	\textbf{CLAIM 2:} \begin{enumerate}
	\item  $\forall h,h' \in \Hist(v_0)$ such that $h \sim h'$: $\sigma_i(h) \sim \sigma_i(h')$ (if $\Last(h) \in V_i$).
	\item  $\forall h \in \Hist(v_0)$, $\forall \tilde{h} \in \widetilde{\Hist}(\Class{v_0})$ such that $h \sim_q \tilde{h}$: $\sigma_i(h) \sim_q \tau_i(\tilde{h})$ (if $\Last(h)\in V_i$).\\
	\end{enumerate}
	\underline{Proof:} \begin{enumerate}\item By $(\star)$, we have that for all $h \sim h' \in \Hist_i(v_0)$ there exists a unique $\tilde{h}$ such that $h \sim_q \tilde{h}$ and $h' \sim_q \tilde{h}$. It follows by construction of $\sigma$ that $\sigma_i(h) \sim_q \tau_i(\tilde{h})$ and $\sigma_i(h') \sim_q \tau_i(\tilde{h})$ and thus $\sigma_i(h) \sim \sigma_i(h')$. It means that $\sigma$ is uniform.
	\item let $h \in \Hist(v_0)$ and $\tilde{h} \in \widetilde{\Hist}(\Class{v_0})$ be two histories such that $\Last(h) \in V_i$ iff $\Last(\tilde{h}) \in \tilde{V}_i$ for some $i \in \Pi$ and such that $h \sim_q \tilde{h}$. By construction of $\sigma$, there exists $\tilde{g}\in \widetilde{\Hist}_i(\Class{v_0})$ such that $h \sim_q \tilde{g}$ and $\sigma_i(h) \sim_q \tau_i(\tilde{g})$. But by $\sim_q$ if $h \sim_q \tilde{g}$ and $h \sim_q \tilde{h}$, then $\tilde{g} = \tilde{h}$. It concludes the proof. \label{proof:proof2it2}
	\end{enumerate}
	 \qed
	\end{minipage}}
	
	By (\ref{proof:proof2it2}) in Claim 2, we have that $\outcome{\sigma}{v_0} \sim_q \outcome{\tau}{\Class{v_0}}$. It follows, due to the fact that $\sim$ respects the gain functions, that $\Gain(\outcome{\sigma}{v_0})= \tGain(\outcome{\tau}{\Class{v_0}}) = p$.\\
	
	Now, we prove that $\sigma$ is an SPE. Ad absurdum, we assume that there exists $hv \in \Hist(v_0)$  , there exists $i \in \Pi$ and there exists $\sigma'_i$ a profitable deviation of $\sigma_{i\restriction h}$ for Player~$i$ in $(\mathcal{G}_{\restriction h},v)$, \emph{i.e.,}
	
	\begin{equation} \Gain_i(h \outcome{\sigma_{\restriction h}}{v}) < \Gain_i(h \outcome{\sigma'_i, \sigma_{-i \restriction h}}{v})\label{eq:proof21}\end{equation}
		
		Let $\tilde{h}\tilde{v} = \Class{h_0}\Class{h_1}\ldots \Class{v}$ with $[h_0]=[v_0]$ we have that $hv \sim_q \tilde{h}\tilde{v}$. By (\ref{proof:proof2it2}) in Claim 2, we have that $h \outcome{\sigma_{\restriction h}}{v} \sim_q \tilde{h}\outcome{\tau_{\restriction \tilde{h}}}{\tilde{v}}$, since $\sim$ respects the gain functions, it follows:

	\begin{equation} \Gain_i(h \outcome{\sigma_{\restriction h}}{v}) = \tGain_i(\tilde{h}\outcome{\tau_{\restriction \tilde{h}}}{\tilde{v}})\label{eq:proof22}\end{equation}

		To obtain the contradiction, we build $\tau'_i$ a profitable deviation of $\tau_{i \restriction \tilde{h}}$ for Player~$i$ in $(\tilde{\mathcal{G}}_{\restriction \tilde{h}}, \tilde{v})$. 
		
		Let $\rho = \outcome{\sigma'_i, \sigma_{i \restriction h}}{v}$, let $\tilde{p} \in \widetilde{\Hist}_i(\tilde{v})$, we define $\tau_i(\tilde{p})$ as follows:

		$ \tau'_i(\tilde{p}) = \begin{cases} [\rho_{n+1}]& \text{ if } \tilde{p} < [\rho_0][\rho_1] \ldots \text{ and } \Last(\tilde{p}) = [\rho_n]\\
											  \text{ some } r \in \Succ(\Last(\tilde{p})) & \text{ otherwise} \end{cases}.$
											  
		Let $\tilde{\rho} = \outcome{\tau'_i, \tau_{-i\restriction \tilde{h}}}{\tilde{v}}$ and let us prove that $\rho \sim_q \tilde{\rho}$, \emph{i.e.,} $\forall n \in \mathbb{N} \rho_n \sim_q \tilde{\rho}_n$. We proceed by induction on $n$.
		
		For $n=0$, $\rho_0= v \sim_q [v] = \tilde{v} = \tilde{\rho}_0$. Let us assume that this property is true for all $n \leq k$ and let us prove it remains true for $n = k+1$.
		
		\begin{itemize}
			\item \textbf{First case: If $\rho_k \in V_i$}, then since $\rho_k \sim_q \tilde{\rho}_k$ by IH, $\tilde{\rho}_k \in \tilde{V}_i$. It follows that:
			
			\begin{align*}
				\rho_{k+1} &\sim_q [\rho_{k+1}] & \text{ by definition of } \sim_q \\
						   &= \tau'_i(\tilde{\rho}_0 \ldots \tilde{\rho}_k) & \text{ By IH, } \rho_0 \ldots \rho_k \sim_q \tilde{\rho}_0 \ldots \tilde{\rho_k} = [\rho_0] \ldots [\rho_k]\\
						   &=\tilde{\rho}_{k+1}
				\end{align*}
				
			\item \textbf{Second case: If $\rho_k \in V_j$} (with $j \neq i$), then as previously $\tilde{\rho}_k \in \tilde{V}_j$. It follows that:
			
			\begin{align*}
				\rho_{k+1} &= \sigma_{j \restriction h} ( \rho_0 \ldots \rho_k) = \sigma_j(h \rho_0 \ldots \rho_k) \\
						   &\sim_q \tau_j(\tilde{h} \tilde{\rho}_0 \ldots \tilde{\rho}_k) & \text{ By IH, } h\rho_0 \ldots \rho_k \sim_q \tilde{h}\tilde{\rho}_0 \ldots \tilde{\rho}_k \text{ and by (\ref{proof:proof2it2}) in Claim 2 }\\
						   &= \tau_{j \restriction \tilde{h}}(\tilde{\rho}_0 \ldots \tilde{\rho}_k) = \tilde{\rho}_{k+1}
			\end{align*}
		\end{itemize}

		Thus, $\rho \sim_q \tilde{\rho}$ and so $h \rho \sim_q \tilde{h} \tilde{\rho}$. Since, $\sim$ respects the gain functions, we can conclude that:
		
	\begin{equation} \Gain_i(h \outcome{\sigma'_i, \sigma_{-i \restriction h}}{v}) = \tGain_i(\tilde{h}\outcome{\tau'_i, \tau_{-i \restriction \tilde{h}}}{\tilde{v}}) \label{eq:proof23}.\end{equation}
		
		By~\eqref{eq:proof21}, \eqref{eq:proof22} and \eqref{eq:proof23}, we can state that:
		
			$$\tGain_i(\tilde{h}\outcome{\tau'_i, \tau_{i\restriction \tilde{h}}}{\tilde{v}}) = \Gain_i(h \outcome{\sigma'_i, \sigma_{-i\restriction h}}{v})
			> \Gain_i(h \outcome{\sigma_{\restriction h}}{v})  = \tGain_i(\tilde{h}\outcome{\tau_{\restriction \tilde{h}}}{\tilde{v}})
			$$ \qed

\end{proof}


\section{Additional material for Section~\ref{section:timedGames}}

In~\cite{brenguier:tel-00827027}, Proposition 6.12 asserts that the value problem for timed games with B\"uchi objectives and only two clocks is EXPTIME-hard. The proof relies on the notion of countdown game~\cite{JurdzinskiLS07} which is known to be EXPTIME-complete. When reading the proof of the latter proposition, one can easily be convinced that it is also proved that the value problem for timed games with reachability objectives and only two clocks is EXPTIME-hard. Indeed, the only accepting state is a deadlock with a self-loop (named $w_\exists$). Moreover, one can also notice that although the results of~\cite{brenguier:tel-00827027} concern concurrent games, the proof of Proposition relies on turn-based games.

The proof of Proposition 6.12 can be slightly modified in order to prove that the constrained existence problem in reachability timed games is EXPTIME-hard with two clocks. The problem in the original proof beeing that Adam does not have a reachability, but a safety objective. Given a countdown game $\mathcal{C}$, we build a reachability timed games by using nearly the same construction as the one presented in the proof of~\cite[Proposition~6.12]{brenguier:tel-00827027}. The difference are the following ones.
\begin{itemize}
\item We replace all the guards $y \ne c_0$ by the guards $y<c_0$.
\item We add a winning state for Adam $w_\forall$. 
\item From every state belonging to Eve, we add a transition to $w_\forall$ with guard $x=0 \wedge y>c_0$.
\end{itemize}
The proposed transformations does not really affect the behaviors of the timed game, in the sense that it still bisimulates closely the countdown game. The only difference is discussed below. In the original encoding, Eve was winning if and only if she is able to reach $w_\exists$. This could happen only when the clock $y$ is equal to $c_0$. As the game is zero-sum, Adam was winning when $w_\exists$ is never reached. In practice, as the timed game of the encoding is strongly non-zeno, in every winning play of Adam, the clock value $y$ eventually overtakes $c_0$. In our new encoding, every winning play of Adam ends up in $w_\forall$.  That is the only difference. This is important, as we can now see the timed game as a reachability time game where both players have a reachability objective. One can be convinced that Eve as a winning strategy (in the original timed game proposed in~\cite{brenguier:tel-00827027}) if and only if there exists an SPE where only Eve achieves her objective (in the variant of the timed game proposed above).

\end{document}